\documentclass[%
 reprint,
 amsmath,amssymb,
 onecolumn,
 showkeys,
 aps,
pra,
]{revtex4-2}

\usepackage{bm}
\usepackage{graphicx}
\usepackage{dcolumn}
\usepackage{bm}
\usepackage{longtable}
\usepackage{soul}
\setstcolor{red}
\usepackage{algorithm}
\usepackage{algpseudocode}
\usepackage[caption=false]{subfig}
\usepackage{tikz}
\usepackage{amsthm}
\newtheorem{theorem}{Theorem}

\newtheorem{example}{Example}
\newcommand{\rblock}[1]{%
  \tikz[baseline=(X.center)]{\node[draw, fill=gray!20, minimum width=1.6cm, minimum height=1.6cm] (X) {$#1$};}%
}

\usepackage{physics}
\usepackage{xcolor,ulem}

\def\code#1{\texttt{#1}}

\begin{document}

\title{A Universal Entanglement Witness Generator}
\author{Aiden R. Rosebush}
\email{aiden.rosebush@mail.utoronto.ca}
\affiliation{
Dept of Electrical \& Computer Engineering, University of Toronto, Toronto, Ontario, Canada M5S 3G4
}

\author{Alexander C. B. Greenwood}
\affiliation{
Dept of Electrical \& Computer Engineering, University of Toronto, Toronto, Ontario, Canada M5S 3G4
}

\author{Andi Shahaj}
\affiliation{
Dept of Electrical \& Computer Engineering, University of Toronto, Toronto, Ontario, Canada M5S 3G4
}

\author{Li Qian}
\affiliation{
Dept of Electrical \& Computer Engineering, University of Toronto, Toronto, Ontario, Canada M5S 3G4
}

\date{\today}

\begin{abstract}
Entanglement witnesses are essential for certifying entanglement, yet constructing ones that are both noise-robust and economical in measurement settings remains challenging—particularly beyond qubits and for non-stabilizer (``magic'') states. We present a machine-learning method that, given a target state and a user-specified number of measurement settings, generates an entanglement witness optimized for noise tolerance in the neighborhood of that state, requiring only local measurements. The approach is fully general, applying to multipartite qubit and qudit systems alike, including non-stabilizer states. For \(N\) qudits of dimension \(d\), we train on the fully-separable eigenstates of each qudit's \(SU(d)\) generators to find a prototype witness, then tune the witness's bias term via gradient descent to maximize noise tolerance. Adversarial training further strengthens the witnesses, delivering greater noise tolerance with even fewer settings; critically, under this scheme the required training-set size becomes independent of system size. We package the entire pipeline as an automated script that, in every case we tested, produces witnesses surpassing all existing methods in noise tolerance and/or number of measurement settings. We demonstrate the method on Bell, GHZ, W, and hypergraph states, along with a range of qudit states, spanning 2–6 qubits, bipartite qudits up to 
d=10, and tripartite qutrits. Our witnesses achieve perfect accuracy across both physical experimental test states and large numerical sets of separable mixed states—including 30 million test states for a 3-qubit W-state witness and 10 million for a 4-qubit hypergraph-state witness—and we experimentally confirm the noise tolerance of Bell- and hypergraph state witnesses on both photonic and superconducting platforms, respectively.
\end{abstract}

\keywords{Entanglement witnesses, entanglement detection, differential programming, machine learning}
\maketitle


\section{\label{sec: Intro} Introduction}


Quantum entanglement is a defining feature of physical systems utilizing quantum technologies. Entanglement is exploited in emerging communication \cite{PhysRevLett.98.060503}, \cite{zhong}, imaging \cite{chen}, and information processing \cite{lukens} technologies. These technologies become more powerful with higher dimensional quantum systems \cite{PhysRevLett.98.060503, zhong, chen, lukens}, such as groups of many quantum 2-level systems (qubits) or systems of many levels (qudits). Advancements in quantum hardware now permit the control of states within increasingly larger dimensions \cite{thomas,kues,imany}. Consequently, efficient state characterization—particularly entanglement detection—is vital for the advancement of high-dimensional quantum technologies.

Entanglement can, in principle, be certified by evaluating an appropriate entanglement measure. While this approach is well established for bipartite systems \cite{PhysRevLett.80.2245,RevModPhys.74.197}, the multipartite setting is more subtle, owing to the absence of a universally accepted notion of maximal entanglement and the diversity of inequivalent classes of multipartite entanglement \cite{plenio2005introduction, acin2001classification}. Although several computable multipartite criteria exist \cite{taming, wei2003geometric,xie2021triangle, coffman2000distributed}, their experimental use often relies on quantum state tomography, which becomes experimentally and computationally complex for appreciably large systems. For a d-level system of N-particles, the number of measurement settings required for a so-called ``tomographically complete" set of measurements grows as $O(d^{2N})$.

A complete characterization is often unnecessary for entanglement detection, though existing alternatives to QST still face significant hurdles. Recent approaches using neural networks \cite{ma} and convex hull approximations \cite{lu} aim to detect entanglement with fewer measurements, yet their efficiency is highly dependent on the system's specifics. For instance, the neural network approach requires  \(O(3^N)\) measurements for $N$-qubit systems and only demonstrates a meaningful reduction when the state is already known to be of the GHZ-type \cite{ma}. Similarly, the convex hull method requires over \(10^3\) binary classifications, for a simple two-qubit case, offering little practical improvement over standard QST. Even methods employing support vector machines \cite{nonlinear_kernel}, with nonlinear kernels, while achieving high accuracy, still rely on a tomographically complete basis, failing to reduce the overall measurement burden.

A more modern deep-learning method for entanglement quantification reaches only five qubits, attains imperfect accuracy, and requires hundreds of measurements \cite{science_advances_DNNs}. A neural-network approach scales considerably further, predicting entanglement metrics such as the Rényi entropy for up to 100 qubits from only local Pauli measurements \cite{100_qubit_ising}. However, it is restricted to the ground and dynamical states of a fixed family of local Hamiltonians—an exponentially small subset of the full state space \cite{poulin_illusion}—and, as a regression-based method, it cannot achieve exact accuracy even in the limit of infinitely many measurements. 

Another approach to entanglement detection, particularly useful for benchmarking quantum technologies \cite{benchmark}, is to construct what is called an entanglement witness. A witness $\mathcal{W}$ is an operator with an expectation value $\langle \mathcal{W}\rangle$ which is non-negative for all separable states and negative only for states that are in proximity to the specific entangled states.  Entanglement witnesses based on local measurements only require $O((d)^{N})$ measurements, which offers superior scaling than all of the other entanglement detection methods mentioned above \cite{PhysRevApplied.19.034058}. Moreover, in practice, witnesses often need far fewer measurement settings, depending on the target state they are set to detect \cite{guhne}.

Witnesses can be compared in terms of the number of measurement settings required and their noise tolerance. In the context of this work, noise tolerance is defined as the maximum amount of white noise the witness can tolerate in a Werner state that is related to the target state \(\rho_e\) through \eqref{eq: noise} while correctly classifying the state as entangled. In general, a better witness should have a higher noise tolerance with fewer measurements, but there is often a trade-off between the two metrics. We can quantify the noise tolerance by the largest value of \(p \in [0,1]\) for which the witness has a negative expectation value when detecting the state given by 

\begin{equation}
\label{eq: noise} 
\rho = (1-p)\rho_e + p\frac{\mathbb{I}}{2^N}
\end{equation}.

We define a measurement setting as a set of measurement operators applied to all qubits in a given system. To evaluate a state using an entanglement witness, the expectation value for each measurement setting must be found using many copies of the state under test. We choose not to define measurement settings by counting the total number of measurements done on all qubits in order to reflect how many copies of a given state are necessary to evaluate it.

While entanglement witnesses using methods such as the stabilizer formalism \cite{toth} and fidelity method \cite{guhne} have been developed, the number of measurements required can still be substantial, depending on the specific context \cite{guhne}. In general, analytical methods like these have rarely been shown to produce witnesses with optimal noise tolerances, and in some cases, have not been shown to require the fewest possible measurements \cite{guhne}. Additionally, while the stabilizer formalism and fidelity method can be applied to many classes of states, it is not straightforward to do so. Often, articles are published on how these methods should be used, justifying dedicated research for detection of particular classes of states like hypergraph states \cite{Ghio_2018} and cluster states \cite{cluster}. 

A complementary approach employs \textit{fully decomposable} witnesses obtained through semidefinite programming \cite{taming}. Given a prescribed set of observables, this method finds an optimal fully decomposable witness within the linear span of these observables.  Importantly, the optimization is not performed directly over arbitrary local measurement settings. In practice, one must first choose a suitable restricted set of observables and subsequently determine how the resulting witness can be evaluated using a feasible and economical collection of compatible local measurement settings. Since several Pauli observables may be obtained from a single local basis measurement, finding an effective decomposition and grouping of the witness terms remains a nontrivial experimental-design problem.

Furthermore, fully decomposable witnesses certify genuine multipartite entanglement only for states that are not so-called ``PPT mixtures" \cite{taming}. Although semidefinite programming can optimize witness coefficients within a specified observable span, it neither guarantees detection of all genuinely multipartite-entangled states nor automatically identifies the most experimentally efficient set of measurement settings.

Choosing an analytical means of constructing a witnesses also forces experimentalists to choose between some witnesses of many measurements, offering high noise tolerances, or witnesses requiring few measurements and with low noise tolerances. Depending on the hardware used to produce the states being detected, using many measurements to detect entanglement might or might not be advantageous. No method exists to find a witness requiring some number of measurements in the wide range between these two values, which could offer a more optimal tradeoff between the expense of measurement and the robustness to noise. 

Construction of entanglement witnesses knowing only the density matrix of a target state also poses problems. Using the stabilizer formalism for mixed states, witnesses can only be constructed if a decomposition into a convex combination of pure states exists, such that those pure states share a set of stabilizing operators \cite{toth}. If no such decomposition exists, a fidelity method witness must be used, which in general tend to require more measurement settings for larger systems \cite{guhne}. The stabilizer formalism can also be applied to qudit states, again requiring knowledge of the class of state and some algebraic properties \cite{szczepaniak2025entanglement}. Overall, analytical methods of deriving entanglement witnesses are non-optimal and inaccessible, while alternative methods could stand to offer far more convenience to experimentalists.

Machine learning (ML) has been used to automate the generation of entanglement witnesses in an effort to address these issues. Support vector machines, or SVMs, are a ML technique commonly used for binary classification of labeled data which have been used for this purpose \cite{PhysRevApplied.19.034058, new_4qubit_ref, rosebush}. These approaches have yielded witnesses with larger noise tolerances than some existing witnesses \cite{PhysRevApplied.19.034058, rosebush} but offer no guarantees on the maximum number of measurement settings, as a user cannot specify how many settings will ultimately be required. These methods offer a convenient means of generating witnesses for specific systems but suffer from the drawback that the user cannot specify how many measurement settings they wish to use for a witness. As a result, they may not compete favourably with analytical methods in terms of fewer measurements or higher noise tolerances. Furthermore these methods require a new computer program to be written for each system a witness might exist in, adding to the difficulty of applying them to arbitrary cases. Additionally, these methods produce numerical, approximate witnesses, which do not necessarily guarantee each witness found will correctly classify all separable states. Our method will find a valid witness whenever one can be found.

In this work we introduce a general, complete framework for constructing witnesses with ML. We no longer use SVMs, as in \cite{PhysRevApplied.19.034058, new_4qubit_ref, rosebush}. Instead, we construct general measurement settings and train them directly, along with coefficients and a bias term. With this method, the user can specify any number of measurement settings with which they wish to construct a witness. We allow for a ML model to be algorithmically generated according to the system to which a target state belongs, and introduce a new differential programming framework which ensures perfect accuracy. Our method can be used for any system of qubits or qudits, subject to available computing resources.

Our approach essentially removes the need for a human to understand the target state they wish to detect, while also providing better results than any analytical method in almost every case. With this method we intend to make entanglement detection an entirely automated process, where any entangled state need only be dealt with indiscriminately as a digital variable. In section \ref{sec: Numerical Verfication} we give examples of how our script can be integrated with additional software to allow for witnesses to be constructed and deployed with zero knowledge.

Our method also overcomes another drawback of existing ML methods, namely, the size of training data required. The least amount of training data for training a witness of an \(N\) qubit system is \(6^N\) separable states \cite{rosebush}, creating memory issues for larger systems. In this work we introduce a second new method which maintains the benefits of universality, choice of the number of measurements, and accuracy, while also employing an adversarial training scheme in which the size of training data no longer depends on the size of the system. This method also produces witnesses requiring even fewer measurements and with higher noise tolerances than any other.

The rest of this work is organized as follows. In section \ref{sec: Overall Algorithm} we introduce the procedure to generate witnesses. In section \ref{sec: Training Measurements} we define general measurement settings mathematically and describe how we train and regularize them. In section \ref{sec: DMSO} we describe our new differential programming scheme which ensures all data is classified correctly. In section \ref{sec: Adverse} we describe our adversarial training method. In section \ref{sec: Examples} we provide some sample results, showing how our methods lead to witnesses with fewer measurements and/or higher noise tolerance in a variety of cases. 

Finally, in \ref{sec: Physical Verification} we describe physical verification of some of our witnesses, photonic qubits in our lab or superconducting qubits on IBMQ hardware, to complement numerical verification of our witnesses in \ref{sec: Numerical Verfication}.

\section{\label{sec: Overall Algorithm} Overall Algorithm}

Our method generates a prototype entanglement witness using machine learning and subsequently refines it via a differential program to correct bias. This process, summarized in Algorithm \ref{alg: main}, comprises three primary stages: basis selection, training, and optimization. Detailed implementations of these steps are provided in Sections \ref{sec: Training Measurements} and \ref{sec: DMSO}.

\textbf{A. Basis Selection and Training Data:}
The algorithm first constructs a basis for Hermitian measurements using the generators of the Lie group $SU(d)$ for $d$-level qudits. These Lambda operators \cite{thew} (which reduce to Pauli operators for $d=2$) allow us to sidestep the existence constraints of mutually unbiased bases—a known limitation in previous ML-derived witnesses \cite{PhysRevApplied.19.034058}.

To generate training data, we use the eigenstates of generalized Pauli operators as the separable set and a single copy of the target state as the entangled set. This approach scales more efficiently than random sampling or other contemporary methods \cite{rosebush}.

\textbf{B. Prototype Witness Generation and Refinement:}
We train a machine learning model to produce a prototype witness, $\mathcal{W}_0$, based on a user-defined number of measurements $M$ ($M \geq 2$ \cite{thew}) and a learning rate $lr$. Training continues until stopping criteria are met. These criteria include a set number of iterations past which the variance in the loss is monitored until it falls below an empirical threshold, ensuring convergence (see Section \ref{subsec: Non-Adversarial Training Algorithm}). While $\mathcal{W}_0$ correctly classifies the training data, it may not possess optimal noise tolerance for the target state.

To resolve this, $\mathcal{W}_0$ is passed to a differential program, \texttt{optimizeMixedSep}. This program identifies the most misclassified separable state $\rho_s$—the state yielding the most negative expectation value under $\mathcal{W}_0$. We then compute the final witness $\mathcal{W}_1$ by shifting the bias term such that $\rho_s$ has an expectation value of exactly zero, ensuring the witness is physically robust.

\begin{algorithm}[H]
\caption{Overall Algorithm}\label{alg: main}
\textbf{Problem Input: }
\begin{itemize}
    \item $N, d, M \in  \mathbb{Z}$
\end{itemize}
\textbf{Output: } 
\begin{itemize}
    \item Entanglement witness $W_1$.
\end{itemize}
\begin{algorithmic}

\Require $N \geq 2$, $d \geq 0$, $M \geq 2$
\begin{enumerate}
    \item Generate separable training data. We use the eigenstates of generalized Pauli operators.
    \item Train prototype witness \(\mathcal{W}_0\). The training algorithm is described in section \ref{sec: Training Measurements}.
    \item Optimize to find most misclassified separable state \(\rho_s\) for the prototype witness. This optimization is described in section \ref{sec: DMSO}.
    \item Use \(\rho_s\) to produce corrected witness \(\mathcal{W}_1\) as $\mathcal{W}_1 \gets \mathcal{W}_0 - Tr(\rho_s \mathcal{W}_0) \mathbb{I}$ where $\mathbb{I}$ is the identity operator for the system.
\end{enumerate}
\end{algorithmic}
\end{algorithm}




\section{\label{sec: Training Measurements} Training Measurements}

In this section, we first introduce a witness model by an ansatz of local measurements, and then train the model using an non-adversarial algorithm to produce the prototype witness. 


\subsection{\label{subsec: Witness Construction} Witness Construction}
In general, quantum observables for discrete systems take the form of Hermitian matrices. We randomly initialize and train a Hermitian matrix for each qubit/qudit and each measurement setting. We denote \(M\) as the number of measurement settings, which can be set arbitrarily, for a system with \(N\) qudits of dimension \(d\). We directly train Hermitian matrices and allow them to take on any general form, instead of parameterizing in terms of any specific basis of measurements. We represent these Hermitian matrices in a single tensor, \(H \in \mathbb{C}^{M \times N \times d \times d}\). 

With each measurement setting, we can infer the value of exactly \(2^N - 2\) additional measurements on subsets of the qudits in our system \cite{thew}. For example, if our model provides one measurement setting for three qubits as \(X \otimes Y \otimes Z\), where \(X, Y\) and \(Z\) represent the Pauli matrices, we can also compute the expectation value for \(X \otimes Y \otimes I\), \(X \otimes I \otimes Z\), \(I \otimes Y \otimes Z\), \(X \otimes I \otimes I\), \(I \otimes Y \otimes I\), \(I \otimes I \otimes Z\) without additional measurements. Here \(I\) represents the identity operator on each qudit, indicating that no measurement is taken. We store these additional measurements in a tensor \(J\), of dimension \(M(2^N -2) \times N \times d \times d\).

We define our witness as a linear combination of observables $H_i$ and $J_i$, with coefficients \(c_i\) for each observable: 

\begin{equation}
    \label{eq: witness summations}
    \mathcal{W}_0 = c_0\mathbb{I} + \sum_{i=1}^M c_i \bigotimes_j H_{ij} + \sum_{i = M+1}^{M(2^N -1)} c_i \bigotimes_j J_{(i-M), j}
\end{equation}.

Note we separated the bias term, i.e., the identity operator, in \eqref{eq: witness summations}. The optimization of coefficient $c_0$ will be discussed in Section \ref{sec: DMSO}. Overall \(c\) has \(1 + M \times (2^N - 1)\) elements.

\subsection{\label{subsec: Loss Function} Loss Function}

 With each iteration, we construct a new witness and evaluate it with a loss function which is similar to the hinge loss typically used for support vector machines \cite{svm_textbook}. We compute the expectation value of our prototype witness in each training step with all of our training data, which includes all eigenstates of the generalized Pauli operators and the target entangled state. We label these predictions as \(y_{pi}\) for each training state \(\rho_{i}\). Since \(\mathcal{W}_0\) is normalized, we know \(-1 \leq y_{pi} \leq 1\). The ground truth label \(y_{ti}\) of each separable state is set to \(+1\), and the ground truth label of the sole entangled target state is set to \(-1\). 

 We compute the loss using 
 \begin{equation}
 \label{eq: loss function}
     l = \max_{i} \{\max(1-y_{ti}y_{pi}, 0)\}.
 \end{equation}.

We consider all training data at once, or full batch training. We take the maximum across the full batch instead of the average, as is common \cite{svm_textbook}, as the model will not classify all states correctly. 


In the next section we summarize the algorithm we use to train the prototype witness \(\mathcal{W}_0\).

\subsection{\label{subsec: Non-Adversarial Training Algorithm} Non-Adversarial Training Algorithm}


As described in section \ref{sec: Overall Algorithm} we use stopping criteria to determine when to cease training. Let \(T\) be the loss variability threshold, defined as the variance of the loss over some number of training steps, called the patience, or \(P\). To generate witnesses with higher noise tolerance, a smaller threshold of variance and larger patience should be used. The variance is only computed after the first \(P\) steps to allow the model time to train. We use a threshold of \(10^{-4}\) and a patience of \(2000\). We chose these values empirically to ensure reasonably good witnesses are produced and that the training converges quickly. With a much smaller threshold, the training continues for a long time and performance does not substantially improve. We increment variable \code{i} to encode the number of training steps.

Let \(l\) be the value of loss function described above in section \ref{subsec: Loss Function}. Also, let \(lr\) represent the learning rate, or a small constant value we use to adjust each trainable variable based on its gradient with respect to the loss.

\begin{algorithm}[H]
\caption{Non Adversarial Training Algorithm}\label{alg: W0 training}

\textbf{Problem Input:}
\begin{itemize}
    \item Entangled target state
    \item Separable training data; pure fully separable eigenstates
\end{itemize}
\textbf{Parameters:} 
\begin{itemize}
    \item User specified number of measurements $M$
    \item User specified number of qubits $N$
    \item User specified dimension $d$
    \item Loss variance threshold \(T\)
    \item Patience \(P\)
\end{itemize}
\textbf{Output:}
Prototype witness $\mathcal{W}_{0}$ for the target state stored in \code{EntRhos}

\vspace{5mm}

\begin{enumerate}
\item{Initialize trainable TensorFlow variables.}
\item{\begin{algorithmic}
\item{\While{Loss variance $\geq$ \(T\) or \code{i} $\leq$ \(P\)}
\begin{enumerate}
\item{Assemble Hermitian measurement matrices \(H\) from trainable variables.}
\item{Assemble inferred measurements \(J\) from measurement matrices.}
\item{Assemble witness \(\mathcal{W}_0\) from all measurement matrices and trainable coefficients.}
\item{Add bias term to witness, create prototype witness \(\mathcal{W}_0\).}
\item{Compute expectation value of all training states with witness \(\mathcal{W}_0\).}
\item{Compute loss \code{l} as the maximum hinge loss across all training states, according to equation \ref{eq: loss function}.}
\If{\code i $>$ \(P\)}
    Compute loss variance over last \(P\) steps.
\EndIf
\item{Update all trainable variables according to gradients with respect to loss \code{l} using automatic differentiation.}
\item{Increment step counter \code{i}.}
\end{enumerate}
\EndWhile}
\end{algorithmic}}
\end{enumerate}

\end{algorithm}

Since the hinge loss aims to optimize the margin on both sides of the witness hyperplane \cite{svm_textbook}, the witness \(\mathcal{W}_0\) must be corrected to ensure perfect accuracy. Additionally, since we train with only a small subset of all separable states but must attain perfect accuracy, we cannot use simply loss functions like those used for weighted SVMS \cite{weight_SVM} or cost sensitive SVMs \cite{cost_SVM} in 
Algorithm \ref{alg: W0 training}. 

To correct \(\mathcal{W}_0\), we find the separable mixed state \(\rho_s\) with the lowest expectation value, which serves as the most misclassified state. This is consistent with our approach in \cite{rosebush}. As discussed in section \ref{sec: Overall Algorithm} and Algorithm \ref{alg: main}, we use another differential program to find \(\rho_s\). We describe this program below in section \ref{sec: DMSO}.

\section{\label{sec: DMSO} Differential Mixed State Optimization}
To optimize over the space of mixed separable states, we consider all permutations of all bipartitions of the complete system. We consider only bipartitions of the system as further partitioning is still possible when each subsystem is represented completely by arbitrary density matrices. Each subsystem could have any degree of separability \cite{guhne}.



%
Recall that a separable state \(\rho\) is defined most generally as one which can be written as a convex combination of product states \cite{guhne}, or 

\begin{equation}
\label{eq: sep def}
    \rho = \sum_i p_i\rho_{i}^{(a)} \otimes \rho_i^{(b)}.
\end{equation}.

Here \(\rho^{(a)}\) and \(\rho^{(b)}\) represent any mixed states in their respective subsystems, entangled or otherwise, meaning that if we parameterize them completely than we need not consider any smaller subsystems.

This parameterization is designed to represent the set of separable states, and we prove
that it carries the minimum number of mixed-state parameters per subsystem necessary to span
this set. This lower bound complements the upper bound on the number of pure states
guaranteed by Caratheodory's theorem \cite{caratheodory, horodecki}. Within this
parameterization we search for the most misclassified separable state,
\begin{equation}
    \label{eq: most misclassified def}
    \rho_s = \arg\min_{\rho}\, \mathrm{Tr}(\rho \mathcal{W}_0),
\end{equation}
for any prototype witness $\mathcal{W}_0$, where the minimization is taken over all separable
states $\rho$. The full parameterization, together with our derivation of the minimum number
of mixed states per subsystem necessary to span the separable set, is provided in Appendix
\ref{sec: DMSO Parameterization}.

Unlike prior work on ML derived entanglement witnesses \cite{PhysRevApplied.19.034058, new_4qubit_ref, rosebush}, our implementation can accommodate any system of qubit or qudits. We algorithmically generate variables to optimize across the space of separable mixed states based on a user specified dimension \(d\) and number of qudits \(N\). We outline this algorithm below and in algorithm \ref{alg: DMSO overall}.

For each bipartiton we also represent additional trainable mixed state subsystems to cover the total number of permutations. In appendix \ref{sec: DMSO Parameterization} we prove that a convex combination of \(v(N, n, d)\) mixed states for each bipartition and permutation is necessary, with

\begin{equation}
\label{eq: v(N, n, d)}
v(N, n, d) = \frac{d^{2(N-n)} + d^{(N-n)}}{2}
\end{equation}

terms in each convex combination and where \(n\) represents the number of qubits in the smaller of the two subsystems in each case. 

Additionally there are exactly \(\binom{N}{n}\) possible permutations of each bipartiton, so we must represent \(\mathcal{N}\) identical separable states as given by  

\begin{equation}
\label{eq: num_copies}
\mathcal{N}(N, n, d) = \binom{N}{n} \times \left(\frac{d^{2(N-n)} + d^{(N-n)}}{2}\right)
\end{equation}

In total we represent \code{S} product states, with 

\begin{equation}
\label{eq: num_product_states}
\text{\code{S}}(N, n, d) = \sum_{n = 1}^{\lfloor N/2 \rfloor}\binom{N}{n} \times \left(\frac{d^{2(N-n)} + d^{(N-n)}}{2}\right)
\end{equation}.

To count the number of trainable parameters, we multiply the number of copies of separable states in equation \ref{eq: num_copies} by the number of parameters in each pair of separable states and sum over all bipartitions. For a state of \(n\) qudits of dimension \(d\) we create \(d^{2n}\) real parameters and \(d^{2n}\) imaginary parameters to form lower triangular matrices $L$ which we assemble into mixed states as given by equation \ref{eq: chol def} below. Overall, a single mixed state for a single bipartition includes  \(d^{2n + 1} + d^{2(N-n) + 1}\) trainable variables. We use an additional array of real valued trainable variables \code{p} to represent the coefficients of density matrices in the final convex combination of product states, per equation \eqref{eq: sep def}. The total number of trainable variables we use for the optimizer is given by

\begin{equation}
\label{eq: total_trainable_variables}
\text{\code{S}}(N, n, d) = \sum_{n = 1}^{\lfloor N/2 \rfloor}\binom{N}{n} \times \left(2d^{2n} + 2d^{2(N-n)} + 1\right) \times \left(\frac{d^{2(N-n)} + d^{(N-n)}}{2}\right).
\end{equation}

The memory usage for the optimizer for some small systems is illustrated in table \ref{tab: memory}. 

\begin{table}[h]
    \centering
    \begin{tabular}{c|c|c}
         \textbf{d} & \textbf{N} & \textbf{Memory Usage}\\
         \hline
         2 & 3 & 9.8 KB\\
         2 & 4 & 189.0 KB\\
         2 & 5 & 3.3 MB\\
         2 & 6 & 62.5 MB\\
         2 & 7 & 1.2 GB\\
         3 & 2 & 3.5 KB\\
         4 & 2 & 10.4 KB\\
         5 & 2 & 24.2 KB\\
         6 & 2 & 48.7 KB\\
         7 & 2 & 88.3 KB\\
         3 & 3 & 195.5 KB\\
    \end{tabular}
    \caption{Optimizer memory usage for small systems. The overall memory usage is higher due to Tensorflow features, intermediate variables, and the size of the software itself.}
    \label{tab: memory}
\end{table}

All the swapping operators between qudits in the system are represented in an array of tensors with the name \code{swaps}. We run the optimization loop for some number of iterations, \code{steps}, typically \(2000\).

We initialize each element of our trainable variables by sampling from a uniform distribution between -1 and 1. This is arbitrary, and we find no difference in results by initializing variables in other ways.

In each step of the training loop, we create states for the subsystems in each distinct bipartition according to the Cholesky decomposition formula \cite{chol}

\begin{equation}
    \label{eq: chol def}
    \rho = \frac{LL^{\dag}}{Tr(LL^{\dag})}
\end{equation}

for a lower triangular matrix \(L\) with real, non-negative diagonal elements.

We repeat this process to produce \code{S}\((N, n, d)\) pairs of density matrices, enough to span the subspace corresponding to every permutation of every bipartition of qubits. We then take their tensor products and apply each swapping operator stored in the array \code{swaps} to the corresponding density matrices. Finally, we can take the convex combination of these density matrices with coefficients \code{p} to produce a final separable state \(\rho_{test}\), but we first constrain each element of \code{p} as 
\begin{equation}
\label{eq: p constrained}
\text{\code{p[i]}} \gets \frac{\text{\code{p[i]}}^2}{\sum_i \text{\code{p[i]}}^2}.
\end{equation}

These density matrices can assume any rank depending on the value of all free parameters. They can be arbitrarily pure or mixed.

At the end of every optimization step we compute the loss as \code{l} $\gets$ \(Tr(\rho_{test} \mathcal{W}_0)\), which is the expectation value of the prototype witness from the measurement training algorithm in section \ref{sec: Training Measurements}. We do one step of gradient descent in order to minimize the loss and find the most misclassified separable state for the witness \(\mathcal{W}_0\).

Similar to our measurement training algorithm, we use a stopping criterion based on the variance in the loss. We define a loss threshold \(t\), typically \(10^{-6}\) and patience \(p\), typically \(2000\) steps. We increment the number of steps completed with counter \(i\). We wait \(p\) steps and then compute the variance in the loss over the last \(p\) steps and stop the optimization when the variance drops below \(t\). The small threshold and long patience ensures the loss has stabilized at a minimum, which helps ensure the optimizer has found the most misclassified separable state up to numerical precision.

The algorithm we use for optimizing across separable states is summarized below in algorithm \ref{alg: DMSO}.

\begin{algorithm}[H]
\caption{\code{optimizeMixedSep(}\(\mathcal{W}_0\)\code{, d, N, steps, lr)}}\label{alg: DMSO overall}
\label{alg: DMSO}
\textbf{Problem Input:}
\begin{itemize}
    \item Number of qudits $N$
    \item Dimension of qudits $d$
    \item Prototype witness $\mathcal{W}_0$
\end{itemize}
\textbf{Parameters:}
\begin{itemize}
    \item Learning rate \code{lr}
    \item Patience \(p\).
    \item Loss variance threshold \(t\).
\end{itemize}
\textbf{Output:}
Most misclassified state of $\mathcal{W}_0$ given by $\rho_s$.

\textbf{Steps:}
\begin{enumerate}
\item{Create trainable pure state density matrices $L$ with nonnegative real diagonal elements. Each matrix $L$ represents systems of up to N/2 qudits.}
\item{Compute all swaps of qudits for all bipartitions of the system.}
\item{
\begin{algorithmic}
\While{Loss variance $\geq$ \(t\) or \(i \leq p\) }
\begin{enumerate}
\item{Assemble separable state $\rho$ from pairs of pure state density matrices $L$, across all permutations of all bipartitions of the system, computing permutations with the swaps.}
\item{Compute expectation value $Tr(\rho \mathcal{W}_0)$.}
\item{Use automatic differentiation to update each variable in each matrix $L$ with learning rate \code{lr}.}
\If{i $\geq$ \(p\)}
    Compute loss variance over last \(p\) steps.
\EndIf
\item{Increment step counter \(i\)}

\end{enumerate}
\EndWhile
\end{algorithmic}}
\end{enumerate}
\end{algorithm}

\section{\label{sec: Adverse} Adversarial Training}
We can extend this method to produce witnesses with higher noise tolerance. Our approach resembles Generative Adversarial Training, in which one model generates new training examples meant to be misclassified by the other \cite{gans}. In this sense the two models compete, as the generator must find more extreme edge cases while the classifier learns to classify all examples correctly. In our case, the generator is our mixed state optimizer presented above in section \ref{sec: DMSO}, while our classifier is the same measurement training framework as described above in \ref{sec: Training Measurements}. We merge the two training algorithms such that we alternate between running one step of the main training loops in the algorithm \ref{alg: W0 training} and procedure described in section \ref{sec: DMSO}. To make progress, we save the latest witness, and all states considered by the optimizer. We also continue to use just one copy of the target entangled state as training data. At the end of the algorithm, we run the separable state optimization again and correct the witness to ensure all separable states have negative expectation value, just as stated in \ref{alg: main}. This step is still necessary because our loss function encourages prototype witnesses with separation between both the entangled and separable training states. 

This approach has two significant benefits and one drawback. As we show in section \ref{sec: Examples}, adversarial training produces witnesses with much higher noise tolerances than training with the separable eigenstates. Additionally, regardless of the size of the system, we can train the models while storing only one new state during each training step, and the necessary number of training steps does not appear to increase dramatically with the size of the system. Instead, the training data size depends on the particular target state and the chosen stopping criterion. The number of training steps matches the number of separable training states and is shown in Figure \ref{fig: epochs} below. 

Effectively, this capability suggests adversarial training allows for an exponential reduction in training data sizes over \cite{rosebush}, down to \(O(1)\). We take advantage of only the most relevant separable training samples and our loss function which only considers the state with the highest hinge loss. The sole drawback is that this approach has longer runtimes than that of the main method presented in \ref{sec: Overall Algorithm}, up to convergence on good witnesses based on the stopping criteria we derived empirically, and in terms of time per training step. For 3, 4 and 5 qubit systems, respectively, the adversarial method takes about 75 ms, 290 ms, and and 2.9 s per training step, whereas the non-adversarial method takes about 53 ms, 181 ms, and 1.8 s per training step. The total runtimes for generating qubit witnesses using the adversarial and non-adversarial methods are shown below in Figure \ref{fig: runtimes}.

For this method, we also use a loss variance threshold and patience as stopping criteria. We use a small threshold, usually \(10^{-6}\), and patience of \(2000\) epochs, which we found empirically through trial and error. These stopping criteria allow for the best witnesses while ensuring convergence within a reasonable amount of time. When this method is close to convergence, the witnesses change less between iterations than witnesses found with the non-adversarial method, due to the training states generated by the optimizer being very close together. Due to this stability, the loss changes less between iterations as well and a smaller loss variance threshold is appropriate. 

\begin{figure}
    \centering
    \includegraphics[width=0.6\linewidth]{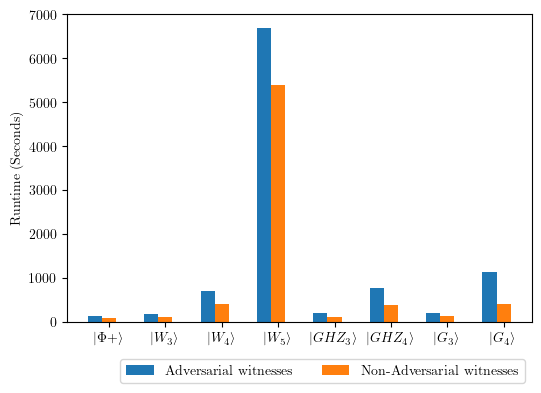}
    \caption{Runtimes, in seconds, of adversarial vs non-adversarial methods.}
    \label{fig: runtimes}
\end{figure}

\begin{figure}
\centering
\includegraphics[width=0.6\linewidth]{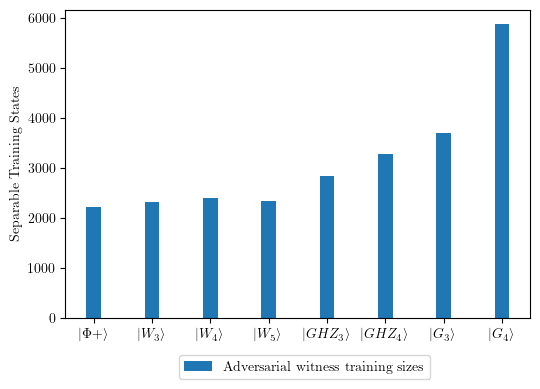}
\caption{Number of separable training states of adversarial qubit witnesses.}
\label{fig: epochs}
\end{figure}

In section \ref{sec: Examples} below we present examples of a wide variety of witnesses for entangled states of qubits and qudits found using both adversarial training and non-adversarial training.

\section{\label{sec: Examples} Examples}

In this section we present the noise tolerances for witnesses found with our methods, given specific numbers of measurements. We first list a range of analytical results for some of the same states, and then present our best witnesses found using both adversarial and non-adversarial training.

For the \(|\Phi^+\rangle\) state, we extended the stabilizer method for finding witnesses for GHZ-type states \cite{toth} down to 2 qubits, for a fair comparison between witnesses with 2 measurement settings, as opposed to 4 global measurement settings required by a witness constructed from the CHSH criterion. 


The state \(|G_N\rangle\) is defined in \cite{Ghio_2018} as an \(N\) qubit hypergraph state with only a single maximum cardinality hyperedge. A maximum cardinality hyperedge means that each qubit is entangled with every other through an extended CNOT gate. Such states can be leveraged to exploit correlations between all qubits in the system and so are useful for common quantum computing algorithms \cite{Ghio_2018}. In general, hypergraph states with a maximum cardinality hyperedge require the most measurement settings for both the fidelity method and stabilizer methods. We chose to compare our method to these analytical methods for the case of a state with a single hyperedge like this for the simplicity preparing the state, the best illustration of the advantages of our method, and for the practical utility of such states in quantum computing.

To construct an analytical witness for an \(N\) qubit hypergraph state with a single maximum cardinality hyperedge, the stabilizer method can construct a witness with \(N\) measurements, while the fidelity method produces a witness with \(\frac{3^N -1}{2}\) measurements \cite{Ghio_2018}. However, the fidelity method produces witness with much higher noise tolerance. 

We also compare to the limits of Werner state noise defined in equation \eqref{eq: noise}, past which states are known to be separable \cite{border_separable, taming}.

We compare analytical qudit witnesses and our generated qudit witnesses below in tables \ref{tab: analytical witnesses}, \ref{tab: adversarial results}, \ref{tab: non-adversarial results} and figure \ref{fig: qubit witness results}. All results  were found using only a MacBook Pro with an M4 Pro Chip and 24 GB of memory. 

\begin{figure}
    \centering
    \includegraphics[width=0.8\linewidth]{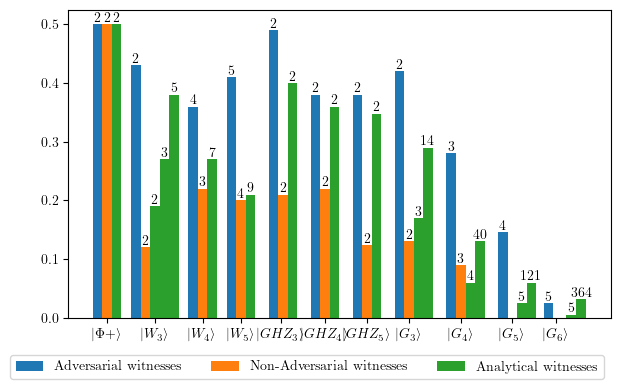}
    \caption{Noise tolerances of analytic and generated qubit witnesses. The number of measurement settings for each witnesss is indicated above each bar. Each group of bars are shown above the label indicating the target state.}
    \label{fig: qubit witness results}
\end{figure}

Figure \ref{fig: qubit witness results} illustrates the limitations of using the eigenstates as training data and highlights the strength of the adversarial training method, particularly with more complicated states like hypergraph states.

\begin{table}[h]
\caption{Qualities of Analytical Witnesses for Qubit States}
\label{tab: analytical witnesses}
\begin{ruledtabular}
\begin{tabular}{c c c c c c c}
\textbf{Target State} & \textbf{d} & \textbf{N} & \textbf{Number of Measurements} & \textbf{Noise Tolerance} & \textbf{Method} & \textbf{Reference}\\
\hline
\textbf{\(\Phi+\)} & 2 & 2 & 2 & 0.50 & Stabilizers & \cite{toth}\\
\textbf{W} & 2 & 3 & 2 & 0.19 & Stabilizers & \cite{toth}\\
\textbf{W} & 2 & 3 & 3 & 0.27 & Stabilizers & \cite{toth}\\
\textbf{W} & 2 & 3 & 5 & 0.38 & Fidelity & \cite{guhne}\\
\textbf{W} & 2 & 4 & 7 & 0.27 & Fidelity & \cite{guhne}\\
\textbf{W} & 2 & 5 & 9 & 0.21 & Fidelity & \cite{guhne}\\
\textbf{GHZ} & 2 & 3 & 2 & 0.40 & Stabilizers & \cite{toth}\\
\textbf{GHZ} & 2 & 4 & 2 & 0.36 & Stabilizers & \cite{toth}\\
\textbf{\(|G_3\rangle\)} & 2 & 3 & 3 & 0.17 & Stabilizers & \cite{Ghio_2018}\\
\textbf{\(|G_3\rangle\)} & 2 & 3 & 14 & 0.29 & Fidelity & \cite{Ghio_2018}\\
\textbf{\(|G_4\rangle\)} & 2 & 4 & 4 & 0.06 & Stabilizers & \cite{Ghio_2018}\\
\textbf{\(|G_4\rangle\)} & 2 & 4 & 40 & 0.13 & Fidelity & \cite{Ghio_2018}\\
\textbf{\(|G_5\rangle\)} & 2 & 5 & 5 & 0.013 & Stabilizers & \cite{Ghio_2018}\\
\textbf{\(|G_5\rangle\)} & 2 & 5 & 121 & 0.06 & Fidelity & \cite{Ghio_2018}\\
\end{tabular}
\end{ruledtabular}
\end{table}

\begin{table}[h]
\caption{Qualities of Witnesses for Qubit States Generated With Adversarial Method}
\label{tab: adversarial results}
\begin{ruledtabular}
\begin{tabular}{c c c c c c}
\textbf{Target State} & \textbf{d} & \textbf{N} & \textbf{Number of Measurements} & \textbf{Noise Tolerance} &\textbf{Noise Tolerance Limit}\\
\hline
\textbf{\(\Phi+\)} & 2 & 2 & 2 & 0.50 & 0.666 \cite{border_separable}\\
\textbf{W} & 2 & 3 & 2 & 0.43 & 0.521 \cite{taming}\\
\textbf{W} & 2 & 4 & 4 & 0.36 & 0.526 \cite{taming}\\
\textbf{W} & 2 & 5 & 5 & 0.41 & -\\
\textbf{GHZ} & 2 & 3 & 2 & 0.49 & 0.571 \cite{taming}\\
\textbf{GHZ} & 2 & 4 & 2 & 0.38 & 0.533 \cite{taming}\\
\textbf{GHZ} & 2 & 5 & 2 & 0.38 & 0.533 \cite{taming}\\
\textbf{\(|G_3\rangle\)} & 2 & 3 & 2 & 0.42 & -\\
\textbf{\(|G_4\rangle\)} & 2 & 4 & 3 & 0.28 & -\\
\textbf{\(|G_5\rangle\)} & 2 & 5 & 4 & 0.15 & -\\
\end{tabular}
\end{ruledtabular}
\end{table}

\begin{table}[h]
\caption{Qualities of Witnesses for Qubit States Generated With Non-Adversarial Method}
\label{tab: non-adversarial results}
\begin{ruledtabular}
\begin{tabular}{c c c c c c}
\textbf{Target State} & \textbf{d} & \textbf{N} & \textbf{Number of Measurements} & \textbf{Noise Tolerance} &\textbf{Noise Tolerance Limit}\\
\hline
\textbf{\(\Phi+\)} & 2 & 2 & 2 & 0.58 & 0.666 \cite{border_separable}\\
\textbf{W} & 2 & 3 & 2 & 0.12 & 0.521 \cite{taming}\\
\textbf{W} & 2 & 4 & 3 & 0.22 & 0.526 \cite{taming}\\
\textbf{W} & 2 & 5 & 4 & 0.20 & -\\
\textbf{GHZ} & 2 & 3 & 2 & 0.21 & 0.571 \cite{taming}\\
\textbf{GHZ} & 2 & 4 & 2 & 0.22 & 0.533 \cite{taming}\\
\textbf{GHZ} & 2 & 5 & 2 & 0.1239 & 0.533 \cite{taming}\\
\textbf{\(|G_3\rangle\)} & 2 & 3 & 2 & 0.13 & -\\
\textbf{\(|G_4\rangle\)} & 2 & 4 & 3 & 0.09 & -\\
\end{tabular}
\end{ruledtabular}
\end{table}

The adversarial training method produces the best results by far, producing higher noise tolerances for W state witnesses with half as many measurements required as required by the fidelity method. It dramatically outperforms hypergraph state stabilizer witnesses, producing much higher noise tolerances than can be found even with the fidelity method while using less than the minimum number of measurements required by the stabilizer method. In particular, our best result is our 5 qubit hypergraph state, where we found a witness with just 4 measurements and more than double the noise tolerance of the best analytical witness requiring 121 measurements. These results are shown in table \ref{tab: adversarial results}. 

Hypergraph states are of particular relevance to quantum information processing \cite{hypergraph_nature}. A recent verification of the presence of a 4 qubit hypergraph state using the fidelity method hypergraph state witness required 256 measurement settings, the same as would be required by full quantum state tomography \cite{hypergraph_nature}, whereas with our witness we would only require 3 measurement settings.

While the non-adversarial method generally produces lower noise tolerances, it still offers witnesses with fewer measurement settings in the hypergraph case and maintains the advantage of shorter runtimes. These results are shown in table \ref{tab: non-adversarial results}. Both methods retain the advantage of universal automation for witness discovery and grant the freedom to choose the number of global measurement settings in each witness.

We note that for both methods, witnesses were found using default learning rates of 0.01 and trained for less than 6000 epochs.

Below we present additional results for both methods with qudit states, particularly those similar to the \(|\Phi^+\rangle\) or GHZ qubit states. These results demonstrate the versatility of our program and offer new witnesses for each target state.

\begin{table}[h]
\caption{Qualities of Witnesses for Qudit States Generated With Adversarial Method}
\label{tab: adversarial results qudits}
\begin{ruledtabular}
\begin{tabular}{c c c c c}
\textbf{Target State} & \textbf{d} & \textbf{N} & \textbf{Number of Measurements} & \textbf{Noise Tolerance} \\
\hline
\textbf{\(\sum_{i = 0}^{2} |ii\rangle\)} & 3 & 2 & 2 & 0.41\\
\textbf{\(\sum_{i = 0}^{3} |ii\rangle\)} & 4 & 2 & 2 & 0.47\\
\textbf{\(\sum_{i = 0}^{4} |ii\rangle\)} & 5 & 2 & 2 & 0.37\\
\textbf{\(\sum_{i = 0}^{9} |ii\rangle\)} & 10 & 2 & 2 & 0.32\\
\textbf{\(|000\rangle + |111\rangle + |222\rangle\)} & 3 & 3 & 2 & 0.33\\
\end{tabular}
\end{ruledtabular}
\end{table}

\begin{table}[h]
\caption{Qualities of Witnesses for Qudit States Generated With Non-Adversarial Method}
\label{tab: non-adversarial results qudits}
\begin{ruledtabular}
\begin{tabular}{c c c c c}
\textbf{Target State} & \textbf{d} & \textbf{N} & \textbf{Number of Measurements} & \textbf{Noise Tolerance} \\
\hline
\textbf{\(|00\rangle + |11\rangle + |22\rangle\)} & 3 & 2 & 3 & 0.37\\
\textbf{\(|00\rangle + |11\rangle + |22\rangle + |33\rangle\)} & 4 & 2 & 3 & 0.11\\
\textbf{\(|00\rangle + |11\rangle + |22\rangle + |33\rangle + |44\rangle\)} & 5 & 2 & 4 & 0.18\\
\textbf{\(|000\rangle + |111\rangle + |222\rangle\)} & 3 & 3 & 3 & 0.13\\
\end{tabular}
\end{ruledtabular}
\end{table}

The individual measurement settings, corresponding to the measurement tensors \(H\) and \(J\), along with coefficients \(c\), can be found in the supplementary material.

\section{\label{sec: Numerical Verfication} Numerical Verification}

\subsection{\label{subsec: test data generation} Numerical Test Data Generation and Testing}
To generate separable testing data, we sampled pure states independently for each susbsystem as the columns of Haar distributed unitary matrices \cite{haar}, and combined them to form mixed separable states according to the same permutations used for the optimizer, mentioned in section \ref{sec: DMSO} and detailed in appendix \ref{sec: DMSO Parameterization}. The coefficients for the convex combination of pure states compose a vector \(\textbf{x} \in \mathbb{R}^K\) defined by a symmetric Dirichlet distribution \cite{PhysRevApplied.19.034058} defined as 
\begin{equation} \label{eq: dirichlet} 
\text{Dir} (x| \alpha) = \frac{\Gamma({K\alpha})}{\left(\Gamma(\alpha)\right)^K} \prod_{j=1}^{K}x_j^{\alpha-1}
\end{equation}

The value \(K\) determines how many pure state density matrices appear in the convex combination. We choose \(K = d^N\), so that the space of states we sample ranges between those that are pure and of full rank. We set our ``concentration parameter" $\alpha$ to be small ($\alpha = 0.01$) to ensure that the states we sample remain high in purity, as shown empirically in \cite{PhysRevApplied.19.034058}. 
High-purity (bi)separable states are closest to the target state (in terms of fidelity, see literature on deriving witnesses using the ``fidelity-method" \cite{PhysRevLett.92.087902}), and are thus ideal candidates for certifying our witnesses near the boundary between (bi)separable and (multipartite) entangled states. 
 This bias towards separable test states closer to the target is reflected in the histogram in Figure \ref{fig: hist} below.

\subsection{\label{subsec: sample histogram} Sample Test Histogram}

For each of the witnesses in section \ref{sec: Examples}, we verified their validity by generating at least 100,000 test states. Figure \ref{fig: hist} presents one sample test histogram, representing the test results of a 3-qubit W state witness, with 2 measurements, generated with the adversarial training method. The states counted with the orange bars represent separable mixed states, and the blue bar count represents Werner states with white noise added from the target state according to equation \ref{eq: noise}. These Werner states include values of \(p\) up to the calculated noise tolerance for the witness.

\begin{figure}[h]
\centering
\includegraphics[width=0.5\linewidth]{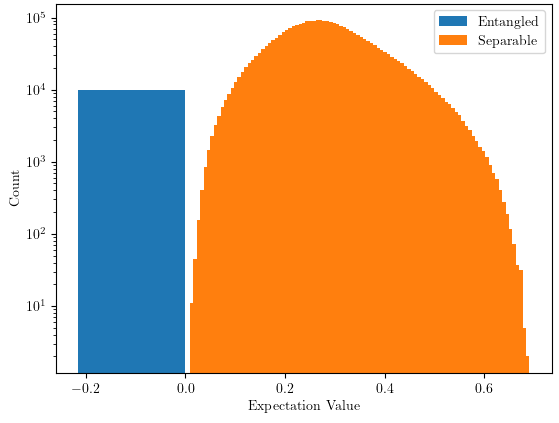}
\caption{Histogram of expectation values of test data for a 3 qubit W state witness with 2 measurement settings generated with our adversarial training method. Shown are the counts of 3 million separable test states (in orange) Werner states based on the target W state up to the maximum noise tolerance of the witness (in blue). Notice the separable states have only positive expectation values, meeting the definition of an entanglement witness. Here we use 100 bins for each type of state, with a binsize of about 0.007.}
\label{fig: hist} 
\end{figure}

\section{\label{sec: Physical Verification} Physical Verification}

\subsection{\label{subsec: Photonic Verification} Photonic Verification}
We verified the witnesses generated by our adversarial approach experimentally using photonic quantum states. We generated a rotated \(|\Psi+\rangle\) state, and used our ML method to generate a witness to detect this state. We show that by physically introducing white noise, we measure a witness expectation value crossing 0 at the same value of \(p\) as we compute by setting the expectation value of the witness with a state given by equation \ref{eq: noise} equal to zero.


 We used entangled photon pairs in the telecommunication band (1530-1600 nm) to realize our target state. To generate our entangled state we use Spontaneous Parametric Down Conversion (SPDC) with a Periodically Poled Silica Fiber \cite{eric's_PPSF_paper}. The experimental set up is depicted in Fig \ref{fig: photonic experiment diagram}. We introduce C-band (1530-1565 nm) and L-band (1570-1600 nm) noise through an Amplified Spontaneous Emission (ASE) source in each band, which we combine through an optical coupler. We chose to add ASE noise due to its broad, flat spectrum, which equally influences both frequency bands and so both qubits. We ensure balanced noise across both bands with a Finisar 4000A Wavershaper. We then combined the SPDC signal with the ASE noise.  We then use a C-L band splitter to separate the entangled photon pair and put each band through a polarization analyzer, each consisting of a Quarter Wave Plate (QWP), Half Wave Plate (HWP), and Polarizer (P). A short description of each of these optical elements and how they are used to implement a projective measurement is included in Appendix \ref{sec: optical elements appendix}. The polarization analyzer enables projective measurements in arbitrary bases, so the expectation value of any operator can be measured. Finally we measure the resulting photon counts through single photon detectors. Our experimental setup is shown below in Figure \ref{fig: photonic experiment diagram}.

Since each local measurement setting for our witness is represented by Hermitian operators, they can be diagonalized and decomposed into projective measurements corresponding to particular eigenstates. We used gradient descent to find waveplate settings which would rotate each eigenstate onto the \(|H\rangle\) state. Using a time interval analyzer, we compared the detection times of infrared photons in the C and L wavelength bands and found that detections coincide with higher likelihood after a particular delay. Since the entangled photons are generated at the same instant in the SPDC process, we expect to see the most coincidences after a fixed delay related to the optical path lengths in our setup. We determined the peak of these coincidences for each measurement and compared the number of counts to determine the relative probabilities of our photonic state being observed in each eigenstate.

\begin{figure}[H]
    \centering
    \includegraphics[width=0.6\linewidth]{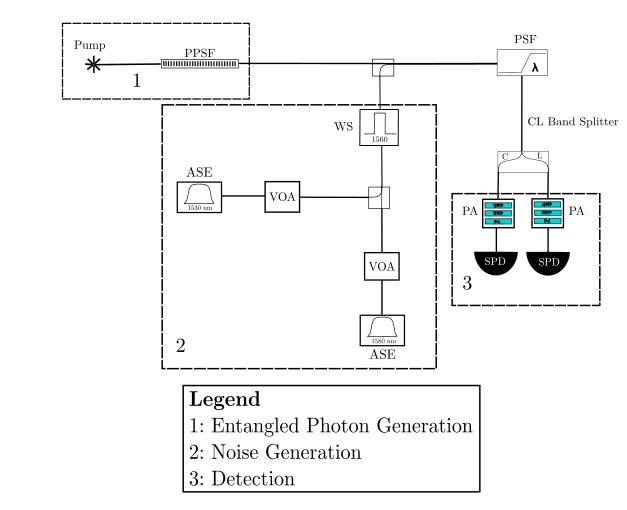}
    \caption{Experimental set up for (1) creating the noisy bipartite entangled photonic state (box 1); (2) adding noise to the entangled state (box 2); and (3) local projective measurements on each photon (box 3). We then combined the signal and the noise with a 50/50 splitter, put the total signal through a Pump Suppression Filter (PSF) before separating the C and L band components for measurement. We used Polarization Analyzers (PAs) consisting of a quarter waveplate, and a half waveplate to project the photons onto the polarizer before the single photon detectors.}
    \label{fig: photonic experiment diagram}
\end{figure} 

We can calculate the noise tolerance of a witness \(\mathcal{W}\) theoretically by setting \(Tr(\rho \mathcal W) = 0\) for \(\rho = p \frac{\mathbb{I}}{d^N} +(1-p)\rho_{t}\) for target state \(\rho_t\). Then 

\begin{equation}
\label{eq: computing theoretical noise tolerance}
    p  = \frac{tr(\rho_t \mathcal{W})}{tr(\mathcal{W}(\rho_t - \frac{\mathbb{I}}{d^N}))}.
\end{equation}

To compute the witness expectation value from experimental results, we need to find the expectation value of each measurement and compute a linear combination of them according to Equation \ref{eq: witness summations}. For the case of $N=2$ and $M=2$, we can write

\begin{equation}
    \label{eq: witness expectation value sum}
    \langle\mathcal{W}\rangle = c_0 + \sum_{i=1}^2 c_i \langle \bigotimes_j H_{ij} \rangle  + \sum_{i = 3}^{6} c_i \langle \bigotimes_j J_{(i-2), j} \rangle = c_0 + \sum_{i=1}^2 c_i \langle H_{i} \rangle  + \sum_{i = 3}^{6} c_i \langle J_{(i-2)} \rangle
\end{equation}

To find the expectation value of \(H_i\), denoted by \(\langle H_i\rangle\), we performed projective measurements onto each eigenstate of \(H_i\), which we call the set \(\{|\Phi_{ik}\rangle\}\). In this case, with two qubits, each global measurement has four eigenstates and requires four projective measurements. The corresponding eigenvalues we denote as \(\{\lambda_{ik}\}\), and those eigenstates have measured detection probabilities \(\{p_{ik}\}\). The probability \(p_{ik}\) of observing a photon pair in each eigenstate is the ratio of the number of photons detected in that state to the total number of photons detected across all eigenstates of the global measurement.

 The expectation value \(\langle H_i\rangle\) is given by 

 \begin{equation}
     \langle H_i\rangle = \sum_{k=1}^4 p_{ik}\lambda_{ik}
 \end{equation}.
 
 We can also infer the expectation values \(\langle J_i \rangle\) of the inferred measurements \(J_i = \bigotimes_j J_{ij}\) in Equation \ref{eq: witness expectation value sum} based on the outcomes of the projective measurements onto the eigenstates \(\{|\Phi_{ik}\rangle\}\). We provide a detailed explanation of how inferred measurement expectation values can be calculated in the 2-qubit case in Appendix \ref{sec: 2 qubit inferred measurement example}.

To compute the expectation value of the witness \(\langle W \rangle\) we took the weighted sum

\begin{equation}
    \langle W \rangle = c_0 + \sum_{i = 1}^{2} c_i \langle H_i \rangle + \sum_{i = 3}^{6} c_i \langle J_{i-2} \rangle.
\end{equation}.

An explanation of how we recover the probabilities \(p_{ik}\) of detecting eigenstate \(k\) of eigenvector \(i\) is included in Appendix \ref{sec: Coincidence histograms}.

Since our measurements are Hermitian operators, they can be diagonalized into projective measurements which form a basis, according to the spectral theorem. The p value is the ratio of noise to total photon flux in an arrival time measurement with no projective elements, the latter of which is equivalent to measuring the expectation value of a projective measurement of every element in a basis onto the target state, as \(Tr(\rho) \equiv \sum_i \langle \Phi_i | \rho | \Phi_i \rangle = 1\). The details of how we obtain \(p\) from our measurement are explained in Appendix \ref{sec: Coincidence histograms}.

We measured the coincidences for two levels of noise added by the ASE noise, shown as the blue and orange data points in Figure \ref{fig: p value vs expectation value}. For each ASE noise level, we can obtain different p-values by including different numbers of noise bins in our calculation, discussed in Appendix \ref{sec: Coincidence histograms}. The overlap between the blue and the orange data points shows that including more bins for noise is equivalent to physically injecting noise into the channel. 

In Figure \ref{fig: p value vs expectation value} below we show that the p value where the witness produces a positive expectation value in the range between 0.4 and 0.5, consistent with our calculated noise tolerance of 0.44. The noise fluxes quoted in the legend refer to the total counts received by both detectors without the presence of the entangled signal photons. Since the two lines closely overlap, we can declare that the witness noise tolerance matches our theoretical prediction in the presence of both physical and synthetic noise. For this original, unique witness, our simulated noise tolerance matches our experiment, meaning this witness is useful for detecting the actual rotated bipartite photonic qubit state we generated. 

\begin{figure}[H]
    \centering
    \includegraphics[width=0.8\linewidth]{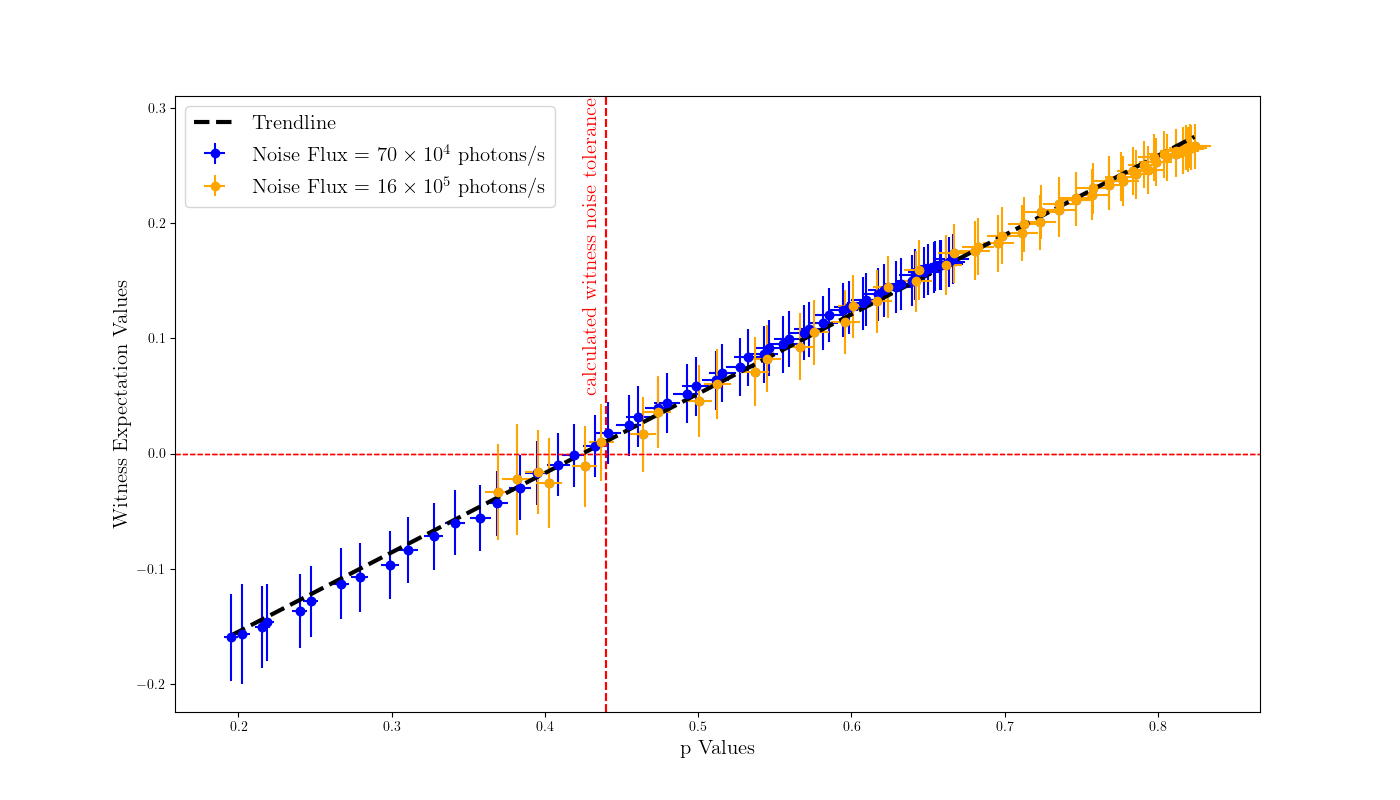}
    \caption{Expectation value of the witness vs p value. We used two values of physical noise, shown in the legend as the two approximate values of photon flux. We simulated many different levels of noise by varying the peak width, or including more noise bins when we calculate the number of coincidences. The black dotted trendline shows a linear fit.}
    \label{fig: p value vs expectation value}
\end{figure}

With this photonic experiment, we confirmed that the measured expectation value our ML derived witness crosses zero close to our calculated noise tolerance. This witness for this physical mixed state behaves as expected with this physical experiment. 

\subsection{\label{subsec: IBMQ Verification} Verification on Superconducting Qubits}
In what follows, we describe experiments on the superconducting device \texttt{ibm\_quebec} where we determine the efficacy of our hypergraph witness when faced with intrinsic decoherence and dephasing of the system; a more in-depth discussion of these phenomena may be found in \cite{PhysRevA.104.022609}. We compare our witness against the existing fidelity method \cite{Ghio_2018}, in terms of both the number of measurement settings required (also presented in Fig. \ref{fig: qubit witness results}) and the relative noise tolerance. We find two important results: (1) our hypergraph witnesses have superior noise tolerance to the fidelity method while requiring fewer measurement settings, and (2) the witness's prediction is in accordance with existing genuine multipartite entanglement measures \cite{taming}. 

Fig. \ref{fig: quantum_circuit} shows the quantum circuit used to compare the two candidate witnesses. We prepare $MN$ identical copies of maximum-cardinality, 3-qubit hyper graph states $G_3$ where $M$ is the number of measurement settings required for a given witness and $N$ is the number of copies required to estimate an expectation value. The circuit used to prepare $G_3$ is marked as the ``state preparation” stage in Fig \ref{fig: quantum_circuit}. After state preparation, we allow the system to idle for a fixed duration, ranging between 0 and 16 $\mu$s, which is used to introduce a controlled amount of noise to the prepared state. The noise injected into the system is a combination of both incoherent and coherent errors. The former of the two results from natural decoherence and dephasing of qubits (characterized by $T_1$ and $T_2$ times, respectively), while the latter arises from ``crosstalk" that introduces unwanted coherences between qubits. As a consequence, the system is \textit{not} expected to follow the white-noise model synthesized in our optical experiments, studied in Section \ref{subsec: Photonic Verification} of this work. After each delay, we perform measurements for both witnesses and estimate their expectation values; the measurement circuits themselves are implemented as local unitary rotations followed by Z-basis measurements. The number of measurement settings required for certifying hypergraph states of maximum cardinality is shown in Table \ref{tab: adversarial results} and Table \ref{tab: non-adversarial results}, using non-adversarial and adversarial training, respectively. We employ witnesses derived from adversarial training that require a total of two measurement settings ($M=2$, see Table \ref{tab: adversarial results}) in this example. In contrast, witnesses derived using the fidelity method require fourteen measurement settings ($M=14$), using the method proposed in \cite{Ghio_2018}.


\begin{figure}
    \centering
    \includegraphics[width=0.7\linewidth]{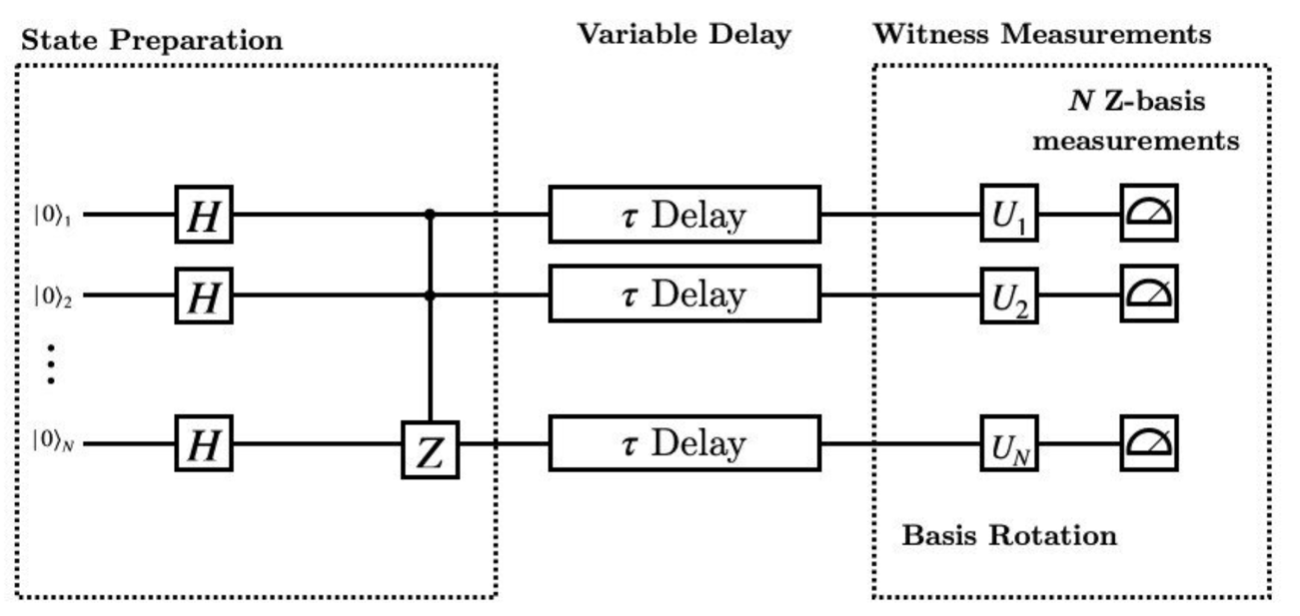}
    \caption{The quantum circuit used to generate hypergraph states $G_N$ of \textit{maximum cardinality} and measure our hypergaph state witness. It includes Hadamard gates on each qubit and a multi-control Z gate. For the purpose of this demonstration, we assume $N=3$.  Delay is added to each qubit to inject a combination of coherent and incoherent errors. Each qubit is then projected onto each eigenstate of each measurement to measure either machine-learning-derived- or fidelity-witness}.
    \label{fig: quantum_circuit}
\end{figure}



Fig. \ref{fig: superconducting expectation values} shows how both fidelity and machine-learning-derived methods predict deterioration of entanglement over a span of 16 $\mu$s. We find that the machine-learning-derived witness offers noise tolerance superior to the fidelity method while remaining in accordance with the expected resource-intensive entanglement of the estimated state $\rho_{est}$.  The ``noise tolerance" of a witness in this experiment is defined as the amount of idle delay required for the expectation value to become positive. We find that while the fidelity-method can only sustain 8 $\mu$s of idle time, our proposed witness can withstand $\sim 11 \;\mu$s while requiring a fraction of the measurement settings (2 vs. 14). The difference in 3 $\mu$s is indeed significant, since most superconducting gate-operations take place on the order of 10s of $n$s on IBM hardware \cite{PhysRevLett.127.080505}.

Following the comparison between fidelity and our machine-learning-derived method, we estimate the genuine multipartite entanglement present in the prepared state after each value of delay. To do so, we performed a 3-qubit, quantum state tomography using software from the Qiskit Experiments software library \cite{kanazawa2023qiskit} to produce an estimated density matrix $\rho_{est}$. The multipartite entanglement possessed by $\rho_{est}$ is determined by evaluating a semi-definite program, also referred to as the ``PPTmixer" measure, as proposed in \cite{taming}. Physically speaking, the entanglement bound represents the lowest possible expectation value attainable by \textit{any fully-decomposable} witness. A witness $\mathcal{W}$ is said to be fully-decomposable if for every subset $M$ it can be expressed as a sum of two positive semidefinite operators $P_M$ and $Q_M$ where $\mathcal{W} = P_M + Q_M^{T_M}$. The superscript $T_M$ indicates the partial transposition across the subsystem $M$ and $\bar{M}$ (the complement of $M$). A full proof of the monotonicity properties for this measure can be found in \cite{taming}. Since the outcome of the PPTmixer measure itself corresponds to a witness expectation value, we expect it to assume a negative value for states that possess genuinely multiparticle entanglement and a positive value otherwise.

The experimental prediction of the PPTmixer measure is indicated by \textit{green} datapoints in Fig. \ref{fig: superconducting expectation values}. We point out that the SDP predicts the presence of entanglement, even after 16 $\mu$s of idle time on the processor. Thus, the degradation in entanglement detected by both witnesses implies coherent errors are the dominant source of ``noise" over ideal periods of this scale. This result is once again in line with the predictions of existing models \cite{PhysRevA.104.022609}, since incoherent errors are introduced to the system at a time scale much longer than the experiments in this work. More specifically, the coherence ($T_1$) and dephasing times ($T_2$) of the transmon qubits used in this experiment are on the order of 300 and 450 $\mu$s respectively.



\begin{figure}
    \centering
    \includegraphics[width=0.5\linewidth]{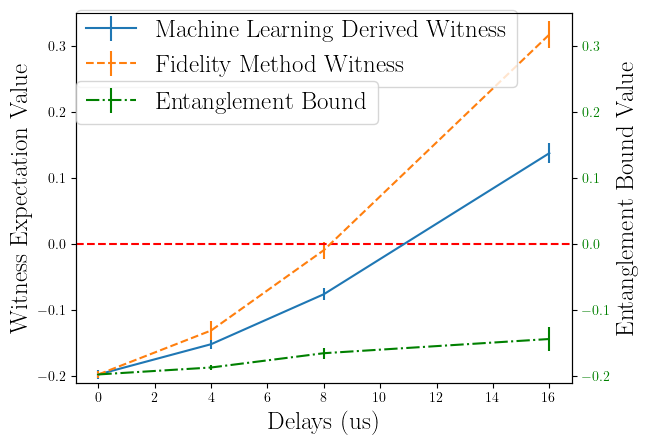}
    \caption{The expectation values of the fidelity witness and our witness with increasing amounts of noise on IBMQ hardware. The entanglement bound (PPTmixer, presented in \cite{taming}) is a measure of entanglement derived from full quantum state tomography to ensure both witnesses are correctly classifying the noisy states. The PPTmixer measure ranges between negative (for states possessing genuine multiparticle entanglement) and positive values. Critically, we find that our machine-learning-derived witness correctly identifies entanglement in instances where the fidelity-method witness fails (between 8 and 11 $\mu$s) while requiring fewer measurement settings ($M=2$ vs. $M=14$).}
    \label{fig: superconducting expectation values}
\end{figure}

\section{\label{Conclusion} Conclusion}

In this work we introduced two new methods for generating entanglement witnesses with ML. These methods are defined by trainable measurement tensors in which a user may specify exactly how many measurements with which they wish to construct a witness. We extended these methods  to create witnesses for arbitrary states of qubits or qudits using the same software script, providing great convenience for experimentalists aiming to benchmark quantum technologies with entanglement witnesses. In one method, we trained the witnesses using the eigenstates of the generalized Pauli operators, while in the other we trained the witnesses using an adversarial scheme in which the witness training algorithm competes with an optimizer searching for misclassified separable states. For both methods we introduced a new optimization method for separable mixed states which we prove in Appendix \ref{subsec: Spanning the set of k-Separable States} can reach every possible separable state. Our methods have been patented under US Patent/Serial Numbers 63/790,389, 19/644,503, and Canadian patent application number 3,307,915.

Finally, we demonstrated the effectiveness of both methods across a variety of qubit and qudit states, showing that in most cases our approaches can find witnesses with fewer measurement settings and/or higher noise tolerances than what is possible by constructing witnesses with the stabilizer formalism and fidelity method. We verified these states numerically with a procedure described in section \ref{subsec: test data generation}, with complete results accessible in the supplementary materials. 

Additionally, we carried our physical verification of several witnesses in section \ref{sec: Physical Verification}, demonstrating that the p value we compute for witness from photonic and superconducting qubit states matches the p value which can be found through experiment.

\section{\label{Acknowledgement} Ackowledgements}
We acknowledge funding from Canada Research Chair program (CRC-2022-00233), NSERC Discovery (RGPIN-2019-07019 and RGPAS-2019-00113), Quantum Alliance Consortium (ALLRP 578462 - 22 and 578460 - 22) Horizon Europe HYPERSPACE Project 101070168, Canada Foundation for Innovation and Ontario Research Fund (CFI-ORF-33415).

\bibliography{b}

\appendix
\section{\label{sec: DMSO Parameterization} Optimizer Parameterization}
\subsection{\label{subsec: Generating Swaps} Generating \code{swaps}}
To represent all possible separable mixed states, we only need to consider biseparable states, as they encompass all k-separable states \cite{acin2001classification,gabriel2010criterion}. To parameterize a biseparable state, we iterate over up to half the number of qubits and divide the system at each step. We represent each bipartition in an array \code{b} which lists all qubits in each susbsystem in two smaller arrays. For example, a bipartition of a 5 qubit system with \(n = 2\) would be represented as \(\text{\code{b}} = [[0,1],[2,3,4]] \) where the qubits are labeled with numbers, and \code{b[0]} represents the smaller bipartiton, and \code{b[1]} represents the larger bipartition.

To compute all possible swapping operators for each system, we first pre-compute all possible swapping operators for two qubits by computing the matrices corresponding to a swap operator applied to the two systems, as described in \cite{rosebush}. Then we generate all the possible pairs of qubits to swap across the bipartition, which is equivalent to choosing any \(n\) qubits from the overall \(N\) qubit system, expressed in terms of the swaps to apply. The algorithm \code{getSwaps(d, N)} iterates over all bipartitions \code{b} in this way and applies another algorithm, \code{getSwitches(b, currentSwitch, allSwitches)}, to compute all the pairs of qubits to search. Finally, \code{getSwaps(d, N)} looks up the corresponding swapping operators, multiplies them together, and returns an array of operators to apply to each product state.

For each bipartiton, we apply the recursive algorithm \code{getSwitches(b, currentSwitch, allSwitches)} to compute the swaps without needing to hard-code any particular system size. This algorithm is detailed below.

\begin{algorithm}[H]
\caption{\code{getSwitches(b, currentSwitch, allSwitches)}}\label{alg: get_switches}
\begin{algorithmic}
\If{\code{b[0] == []} or \code{b[1] == []}}: \Comment{Base case.}
        \State \Return allSwitches
\Else
        \For{i $\leq n$}:
            \For{j $\leq N-n$}:
                \State currentSwitchTemp $\gets$ currentSwitch + \code{[[b[0][i],b[1][j]]]}
                \State allSwitches += [currentSwitchTemp]
                \State bReduced $\gets$ \code{[b[0][i+1:], b[1][j+1:]]} \Comment{To avoid duplicates, consider only qubits which haven't yet been swapped for the recursive step.}
                \State allSwitches $\gets$ \code{getSwitches(bReduced, currentSwitchTemp, allSwitches)}
            \EndFor
        \EndFor
        \State \Return allSwitches
\EndIf
\end{algorithmic}
\end{algorithm}
\subsection{\label{subsec: Spanning the set of k-Separable States} Minimum Number of Mixed States Necessary to Span the Biseparable Set}

In this appendix, we consider the ``coverage" of our parameterization of the set of biseparable states, described in the main text of this work. We argue that a necessary but \textit{not sufficient} condition for parameterizing the biseparable states $\rho^{(BS)}\in\mathcal{B}$ depends on the linear independence between blocks when expressing $\rho^{(BS)}$ in the block representation.
For the purpose of this argument, we define $\rho^{(BS)}$ as a classical mixture of states belonging to all possible bipartitions. Consider, for instance, the three-qubit case where
\begin{equation}
    \rho^{(BS)} = p^{(A|BC)}\rho^{(A|BC)} + p^{(B|AC)}\rho^{(B|AC)} + p^{(C|AB)}\rho^{(C|AB)},
    \label{eq:general_bisep_state}
\end{equation}
where $p^{(\alpha|\beta \gamma)}$ and $\rho^{(\alpha|\beta \gamma)}$ represent the weights and states of the $(\alpha|\beta \gamma)$ partititions, respectively. We use the notation $(\alpha|\beta\gamma)$ to represent a bipartition where qubit $\alpha$ is separable from the bipartite system consisting of $\beta$ and $\gamma$ qubits.
In this section, we consider the parameterization of a single bipartition, $\rho = \rho^{(\alpha|\beta\gamma)}$. We refer to this state as
\begin{equation}
    \rho = \sum_i^S p_i\rho^{(a)}_i\otimes \rho^{(b)}_i,
    \label{eq:biseparable_decomposition}
\end{equation}
where $i$ is an index used to identify terms in a convex combination. The indices `a' and `b' represent subdivisions of a given bipartition. For example, a composite system of qubit `A' and the joint Hilbert space of qubits `B' and `C', as one possible bipartition shown in \eqref{eq:general_bisep_state}, are an example of one such subdivision. The task of sampling each individual state $\rho^{(a)}_i$ and $\rho^{(b)}_i$ is well-understood and often performed using Cholesky decomposition \cite{horn2012matrix,PhysRevA.61.010304, PhysRevLett.127.140502}, as will be shown later. However, determining the \textit{number of terms} $S$ in a convex combination is a non-trivial task. In the context of decomposition into arbitrary \textit{pure} states in a convex combination, the upper bound in terms is decided by Carathéodory's theorem \cite{caratheodory, horodecki,zhu2025unified}, which predicts an exponential growth in number of weighted terms required for a convex combination. As the authors of \cite{zhu2025unified} point out, it is rarely necessary to saturate Carathéodory's bound, and this often results in unneccessary computational overhead. One could also consider the \textit{Hermitian operator Schmidt Rank}, described in \cite{de2019separability}, to determine the number of tensor product terms of Hermitian operators. Unlike the Carathéodory bound, the Hermitian operator Schmidt rank does not straightforwardly apply to products of physical density matrices. Each state in the decomposition \eqref{eq:biseparable_decomposition}, referred to as $\rho_{i}^{(l)}$, can be sampled using the well-known Cholesky decomposition \cite{horn2012matrix}: 

\begin{equation}
    \rho_{i}^{(l)} = \frac{L_i^{(l)}\left(L_i^{(l)}\right)^\dagger}{\Tr \left[L_i^{(l)}\left(L_i^{(l)}\right)^\dagger\right]},
\end{equation}

where $L_i^{(l)}$ is a lower triangular matrix with complex entries (with the exception of real nonnegative diagonal elements) and dimensions $n_l\times n_l$, where $l$ is an index that corresponds to partition $a$ or $b$ ($l\in\{a,b\}$). 
Notably, every density matrix \(\rho\) has at least one Cholesky decomposition defined by \(L\) \cite{chol}.
In the original optimizer \cite{rosebush}, we represented pure states directly, and considered all possible partitions of an \(N\) qubit state, down to 1 and 2 qubit systems. For instance, to represent a 5 qubit state, we would consider the partition with two groups of 2 qubit systems and one separate qubit system.
If we set \(L\) arbitrarily, then we can construct all possible states \(\rho\) which account for k-separability and entanglement within \(\rho\). This is the key insight which makes it possible to construct a k-separable state with only bipartitions. A state or permutation of \(\rho^{(a)} \otimes \rho^{(b)}\) then accounts for all possible states or permutations of the tensor product between any mixed state in each subsystem, represented as \(\left( \sum_i p_i^{(a)} \bigotimes_j \rho^{(a)}_{ij}\right)\otimes \left( \sum_i p_i^{(b)} \bigotimes_j \rho^{(b)}_{ij} \right)\). A state like \(\rho^{(a)} \otimes \rho^{(b)}\) does not necessarily account for all states which can be reached by some convex combination of mixed states in each subsystem, which could instead be written as \(\sum_i p_i \rho^{(a)}_i \otimes \rho^{(b)}_i\). Each \(\rho^{(a)}_i\) and \(\rho^{(b)}_i\) is defined to be freely varying, so it can represent any mixed state for its subsystem; we encode these with trainable variables in Tensorflow, adjusted to take on any value corresponding to a valid density matrix via the Cholesky decomposition.

We now state the necessary condition on the number of terms $S$ needed for the blocks of $\rho$ to admit linear independence, and hence for the parameterization to be capable of spanning the corresponding bipartition.

\begin{theorem}[Minimum number of terms for a biseparable bipartition]
\label{thm:min_terms}
Let
\begin{equation}
\rho = \sum_{i=1}^S p_i\, \rho^{(a)}_i \otimes \rho^{(b)}_i,
\qquad \rho^{(a)}_i \in \mathbb{C}^{n_1\times n_1},\ \ \rho^{(b)}_i \in \mathbb{C}^{n_2\times n_2},
\end{equation}
and write $\rho$ in the block form of \eqref{eq:tensor_product}, with blocks $\bar\rho_{jk} = \sum_{i=1}^S p_i \rho^{(a)}_{i,jk}\rho^{(b)}_i$. If
\begin{equation}
S < \max\left\{\frac{n_1^2+n_1}{2},\ \frac{n_2^2+n_2}{2}\right\},
\end{equation}
then the lower-triangular blocks $\{\bar\rho_{jk}\}_{j\geq k}$ are necessarily linearly dependent, and the parameterization cannot span this bipartition. Consequently,
\begin{equation}
S \;\geq\; \max\left\{\frac{n_1^2+n_1}{2},\ \frac{n_2^2+n_2}{2}\right\}
\end{equation}
is a necessary condition for the blocks to be linearly independent.
\end{theorem}

\begin{proof}
The state $\rho$ can be expressed in the block form
\begin{equation}
\rho = \sum_i^S p_i \rho_{i}^{(a)} \otimes \rho_{i}^{(b)} =
\begin{pmatrix}
\sum_i^S p_i \rho_{i,11}^{(a)}\rho_{i}^{(b)} & \sum_i^S p_i \rho_{i,12}^{(a)}\rho_{i}^{(b)} & \dots \\
\sum_i^S p_i \rho_{i,21}^{(a)}\rho_{i}^{(b)}& \sum_i^S p_i \rho_{i,22}^{(a)}\rho_{i}^{(b)} & \dots \\
\vdots & \vdots & \ddots\\
\end{pmatrix},
\label{eq:tensor_product}
\end{equation}
where $\rho_{i,jk}^{(l)}$, with $l\in\{a,b\}$, denotes the $i$th state in the decomposition, restricted to the $l$th part of the bipartition, evaluated at row $j$ and column $k$. This block form can be illustrated succinctly:
\begin{equation}
    \rho =
    \left(
    \begin{array}{cccc}
        \rblock{\bar{\rho}_{11}}     & \rblock{\bar{\rho}_{12}}     & \cdots & \rblock{\bar{\rho}_{1n_1}}   \\[3em]
        \rblock{\bar{\rho}_{21}}     & \rblock{\bar{\rho}_{22}}     & \cdots & \rblock{\bar{\rho}_{2n_1}}   \\[3em]
        \vdots                      & \vdots                      & \ddots & \vdots                       \\[1em]
        \rblock{\bar{\rho}_{n_11}}   & \rblock{\bar{\rho}_{n_12}}   & \cdots & \rblock{\bar{\rho}_{n_1n_1}}
    \end{array}
    \right),
    \qquad \bar{\rho}_{jk}\in\mathbb{C}^{n_2\times n_2},
\end{equation}
where each block satisfies
\begin{equation}
    \bar{\rho}_{jk} = \sum_{i=1}^{S} p_i \rho_{i,jk}^{(a)}\rho_i^{(b)}.
    \label{eq:block_decomposition}
\end{equation}
Hence each block $\bar{\rho}_{jk}$ is a complex linear combination of the matrices $\rho_i^{(b)}$, so
\begin{equation}
    \bar{\rho}_{jk} \in \operatorname{span}_{\mathbb{C}}\left\{ \rho_1^{(b)},\ldots,\rho_S^{(b)} \right\},
    \label{eq:block_span}
\end{equation}
and consequently
\begin{equation}
    \operatorname{span}_{\mathbb{C}}\left\{\bar{\rho}\right\}_{jk}^{d_A}
    \subseteq
    \operatorname{span}_{\mathbb{C}}\left\{ \rho_1^{(b)},\ldots,\rho_S^{(b)} \right\},
\end{equation}
from which
\begin{equation}
    \dim\left\{\operatorname{span}_{\mathbb{C}}\left\{\bar{\rho}\right\}_{jk}^{d_A}\right\}
    \leq
    \dim\left\{\operatorname{span}_{\mathbb{C}}\left\{ \rho_1^{(b)},\ldots,\rho_S^{(b)} \right\}\right\} \leq S.
\end{equation}
The left-hand dimension is maximized precisely when the blocks $\bar\rho_{jk}$ are linearly independent, so it remains to find the smallest $S$ for which this can occur.

\emph{Case 1: $n_1 \geq n_2$.} Since $\rho$ is Hermitian, $\bar\rho_{jk} = \bar\rho_{kj}^{\dagger}$, so we restrict attention to the lower-triangular blocks $j \geq k$. Suppose these were linearly dependent; then there exist complex coefficients $a_{jk}$, not all zero, with
\begin{equation}
    \sum_{j \geq k} a_{jk} \sum_i^S p_i \rho^{(a)}_{i,jk}\rho^{(b)}_i = 0.
    \label{eq: 7}
\end{equation}
Exchanging the order of summation,
\begin{equation}
    \sum_{i}^S \rho^{(b)}_i \sum_{j \geq k} a_{jk} p_i \rho^{(a)}_{i,jk} = 0.
    \label{eq: 8}
\end{equation}
If there exist coefficients $\{b_{jk}\}$, not all zero, such that
\begin{equation}
    \sum_{j \geq k} b_{jk} p_i \rho^{(a)}_{i,jk} = 0 \quad \text{for every } i,
    \label{eq: 9}
\end{equation}
then setting $a_{jk}=b_{jk}$ in \eqref{eq: 8} forces the sum to zero, giving \eqref{eq: 7}. Arranging, for fixed $(j,k)$, the values $p_i\rho^{(a)}_{i,jk}$ over $i=1,\ldots,S$ into a vector
\begin{equation}
\nu_{jk} =
\begin{pmatrix}
p_1 \rho^{(a)}_{1,jk}\\
p_2 \rho^{(a)}_{2,jk}\\
\vdots\\
p_S \rho^{(a)}_{S,jk}
\end{pmatrix},
\label{eq: 10}
\end{equation}
condition \eqref{eq: 9} holding is equivalent to the vectors $\{\nu_{jk}\}_{j\geq k}$ being linearly dependent. So the lower-triangular blocks $\bar\rho_{jk}$ are linearly independent only if the vectors $\nu_{jk}$ are linearly independent, i.e., only if the matrix $\nu$ with elements $\rho^{(a)}_{i,jk}$ has linearly independent columns. Since $\rho^{(a)}_i \in \mathbb{C}^{n_1\times n_1}$ is Hermitian, the number of columns of $\nu$ (one per lower-triangular entry) is $\frac{n_1^2+n_1}{2}$, and the rank of $\nu$ is bounded by $\min\left\{\frac{n_1^2+n_1}{2}, S\right\}$. Linear independence of the columns therefore requires
\begin{equation}
S \geq \frac{n_1^2+n_1}{2}.
\end{equation}

\emph{Case 2: $n_2 > n_1$.} Since scalar multiplication is commutative, $\sum_i^S p_i \rho^{(a)}_i \otimes \rho^{(b)}_i$ has exactly the same elements as $\sum_i^S p_i \rho^{(b)}_i \otimes \rho^{(a)}_i$, only rearranged into blocks of the other size. The argument of Case 1 therefore applies with the roles of $\rho^{(a)}_i$ and $\rho^{(b)}_i$ exchanged, giving
\begin{equation}
S \geq \frac{n_2^2+n_2}{2}.
\end{equation}

Combining both cases, the blocks $\bar\rho_{jk}$ can be linearly independent only if
\begin{equation}
S \geq \max\left\{\frac{n_1^2+n_1}{2},\ \frac{n_2^2+n_2}{2}\right\},
\end{equation}
which proves the claim.
\end{proof}

Since each $\rho^{(a)}_i$ and $\rho^{(b)}_i$ varies freely through the Cholesky decomposition, the bound of Theorem~\ref{thm:min_terms} is also the minimum number of mixed states per subsystem necessary to remove the forced linear dependence among the blocks — that is, the minimum $S$ necessary for the parameterization to span the set of biseparable states for a given bipartition.

A simple example is a 2 qubit mixed separable state which cannot be represented with \(S < 3\), defined below.

\begin{example}
The mixed separable state $\rho$ given by
\begin{equation}
    \label{eq: optimizer appendix example}
    \rho = \frac{1}{3} \begin{pmatrix} 1 & 0 \\ 0 & 0 \end{pmatrix} \otimes \begin{pmatrix} 1 & 0 \\ 0 & 0 \end{pmatrix} + \frac{1}{3} \begin{pmatrix} 0 & 0 \\ 0 & 1 \end{pmatrix} \otimes \begin{pmatrix} 0 & 0 \\ 0 & 1 \end{pmatrix} + \frac{1}{12}\begin{pmatrix} 1 & 1 \\ 1 & 1 \end{pmatrix} \otimes \begin{pmatrix} 1 & 1 \\ 1 & 1 \end{pmatrix}
    =
    \frac{1}{12}
    \begin{pmatrix}
        5 & 1 & 1 & 1\\
        1 & 1 & 1 & 1\\
        1 & 1 & 1 & 1\\
        1 & 1 & 1 & 5\\
    \end{pmatrix}
\end{equation}
cannot be represented by any convex combination of 2-qubit separable states of the form
\begin{equation}
    \label{eq: optimizer appendix convex combination redef}
    \rho_S = \sum_{i=0}^{S-1} p_i \rho^{(a)}_i \otimes \rho^{(b)}_i
\end{equation}
for $S < 3$.
\end{example}





\section{\label{sec: 2 qubit inferred measurement example} Calculating Inferred Measurements for 2 Qubit Photonic Witness Measurement}
We then computed the expectation value as follows. A single measurement setting \(H_i\) is composed of two single qubit operators \(H_{i0}\) and \(H_{i1}\), with \(H_i = H_{i0} \otimes H_{i1}\). Each local operator has eigenvalues \(\lambda_{ij0}\) and \(\lambda_{ij1}\), and by linearity of the tensor product, the overall measurement \(H_i\) has global eigenvalues \(\lambda_{ik} = \lambda_{ij0} \times \lambda_{ij1}\) where \(k \in [0,3]\). To match, each eigenvalue \(\lambda_{ijk}\) has corresponding eigenstates \(|\Phi_{ijk}\rangle\) such that the overall measurement \(H_{i}\) has four eigenstates: \(|\Phi_{i00}\rangle \otimes |\Phi_{i10}\rangle\), \(|\Phi_{i00}\rangle \otimes |\Phi_{i11}\rangle\), \(|\Phi_{i01}\rangle \otimes |\Phi_{i10}\rangle\), and \(|\Phi_{i01}\rangle \otimes |\Phi_{i11}\rangle\). We measured the probability \(p_i\) of our photon pairs being observed in each eigenstate as \(p_{ik}\) where for example \(p_{03}\) is the probability of measuring a photon for measurement 0 in the eigenstate  \(|\Phi_{001}\rangle \otimes |\Phi_{011}\rangle\). We computed the probabilities of inferred single qubit measurements with \(\{p_{ik}\}\) as \(\boldsymbol{p}_{l} = p_{ik_1} + p_{ik_2}\) with \(k_1 != k_2\) where we use index \(l\) to match the witness coefficient index and the sum of probabilities reflects the probability of a single qubit being found in one of the eigenstates for just that operator on that qubit. For example, we could define the probability of one photon being observed in eigenstate \(|\Phi_{i00}\rangle\) for measurement \(i\) as the sum of the probabilities of both photons being observed as \(|\Phi_{i00}\rangle \otimes |\Phi_{i10}\rangle\) or as \(|\Phi_{i00}\rangle \otimes |\Phi_{i11}\rangle\).

\section{\label{sec: optical elements appendix} Projective Measurements Using the Polarization Analyzer Elements}
In Section \ref{subsec: Photonic Verification} we describe representing quantum states with physical single photon polarization, and measuring those states using entanglement witnesses which we break down into a set of projective measurements. We carry out those projective measurements by rotating particular eigenstates related to the witness onto the horizontally polarized state. To change the polarization of these photons we use optical elements called wave-plates.

Wave-plates take advantage of birefringence, a physical property of certain materials in which the index of refraction is different for electric field oscillations along different axes. Polarized light will experience different optical path lengths along these axes, introducing phase differences. Two basic wave-plates are a half wave-plate and quarter wave-plate, which introduce phase differences between perpendicular axes of \(\pi\) and \(\pi/2\), respectively 
\cite{pedrotti}. 

It can be shown that a half-wave plate, rotated at an angle \(\theta_0\) to the horizontal, has the effect of a unitary operator written as \cite{jones_matrices}

\begin{equation}
    \label{eq: HWP}
    U_{HWP} = 
\begin{pmatrix}
\cos^{2}\theta_0 - \sin^{2}\theta_0 & 2\cos\theta_0\sin\theta_0 \\
2\cos\theta_0\sin\theta_0& \sin^{2}\theta_0 - \cos^{2}\theta_0
\end{pmatrix}
\end{equation}.

A quarter wave-plate rotated at an angle \(\theta_1\) to the horizontal produces a unitary operation which can be written as \cite{jones_matrices}

\begin{equation}
    \label{eq: QWP}
    U_{QWP} = 
\begin{pmatrix}
\cos^{2}\theta_1+ i\sin^{2}\theta_1 & (1-i)\sin\theta_1\cos\theta_1 \\
(1-i)\sin\theta_1\cos\theta_1 & \sin^{2}\theta_1 + i\cos^{2}\theta_1
\end{pmatrix}
\end{equation}

We also use polarizers, set to admit only horizontally polarized light. Polarizers achieve this by using anisotropic materials, like long-chain polymers or crystals—whose internal structure aligns electrons in a specific direction \cite{pedrotti}. Horizontal polarizers act as non-unitary operators, with a Jones matrix given by 

\begin{equation}
    \label{eq: POL}
    U_{POL} = 
\begin{pmatrix}
1 & 0 \\
0 & 0
\end{pmatrix}
\end{equation}

Our polarization analyzer uses a quarter wave-plate followed by a half wave-plate followed by a polarizer. For he eigenstate for measurement \(i\), qudit \(j\) and outcome \(k\), which we denote \(|\Phi_{ijk}\rangle\), We use gradient descent to find the parameters \(\theta_0\) and \(\theta_1\) such that

\begin{equation}
    \label{eq: QWP HWP Pol expression}
    U|\Phi_{ijk}\rangle \equiv U_{HWP}U_{QWP}|\Phi_{ijk}\rangle = |H\rangle
\end{equation}.

Projecting an unknown state \(|\Psi\rangle\) onto the eigenstate \(|\Phi_{ijk}\rangle\) is equivalent to rotating \(|\Psi\rangle\) by  \(U_{HWP}U_{QWP}\) and then projecting it onto the horizontally polarized \(|H\rangle\) state, as shown below:

\begin{equation}
    \label{eq: projection equivalence}
    \langle \Psi |\Phi_{ijk}\rangle = \langle \Psi|U^\dagger U|\Phi_{ijk}\rangle = \langle \Psi|U^\dagger |H\rangle
\end{equation}.

Then, when the polarizer \(U_{POL}\) is applied, only photons in the state \(|H\rangle\) is admitted to the single photon detectors, corresponding to photons originally in the eigenstate \(|\Phi_{ijk}\rangle\) we wish to detect. With this two qubit system, and two detectors, we can then detect each eigenstate of global measurement \(i\) given by \(|\Phi_{ik}\rangle = |\Phi_{i0k}\rangle \otimes |\Phi_{i1k}\rangle\).

Only a quarter wave-plate and a half-wave plate are necessary to rotate an arbitrarily polarized state onto any linearly polarized state, including \(|H\rangle\) \cite{fano}.

\section{\label{sec: Coincidence histograms} Calculating \(p_{ik}\) and \(p\) from Coincidence Histograms}

We construct coincidence histograms by performing repeated projective measurements on a bi-partite Werner state, defined in Eq. (\ref{eq: noise}), using the method described in Section \ref{subsec: Photonic Verification}. Each histogram records the accumulated counts sorted by relative time of arrival between photon detection events. Examples of these histograms are shown in Fig. \ref{fig: cumulative coincidences} a) and b). The histogram has a binsize of 328 picoseconds. 

We find the probabilities \(p_{ik}\) by evaluating the histograms of coincidences for each projective measurement. We choose a certain number of bins \(n\) around the peak of the histogram, called the coincidence window (the gray region in the figure), to count the coincidences. For one measurement, we compute \(p_{ik} = A_{ik}/(\sum_q A_{iq})\), where \(A_{ik}\) is number of coincidences within the window.

We sum the counts from individual histograms to produce the overall coincidence histograms shown in Fig. \ref{fig: cumulative coincidences} c) and d). We use these overall histogram to compute the p values.  The average noise level \(B\) is indicated by a dashed red line in the Figure, and we sum all individual coincidence totals \(A_{ik}\) for a given coincidence window, as indicated by the gray band. For total coincidence counts \(\sum_{ik} A_{ik}\) and total noise counts within the coincidence window \(nB\), we compute \(p\) in Eq. (\ref{eq: noise}) as \(nB/\sum_{ik} A_{ik}\). 

We can vary the level of noise in the Werner state ``seen" by our witness in two ways. One is to physically inject more noise into the channel, as described by Section \ref{subsec: Photonic Verification}. The other is to include more noise by increasing the number of bins \(n\) of the coincidence window. We implemented both in our witness expectation value calculation which we present in Fig. \ref{fig: p value vs expectation value}.



\begin{figure*}[htbp] 
    \centering
    \subfloat[Coincidence histogram for a particular projective measurement with noise flux of \(70 \times 10^4\) photons per second. The gray region shows the noise bins. \label{fig: low noise single coinc histo}]{%
        \includegraphics[width=0.48\textwidth]{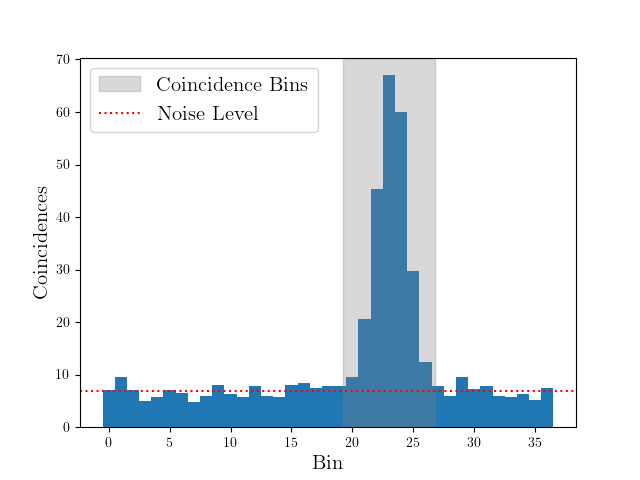}%
    }
    \subfloat[Coincidence histogram for a particular projective measurement with noise flux of \(16 \times 10^5\) photons per second. The gray region shows the noise bins. \label{fig: high noise since coinc histo}]{%
        \includegraphics[width=0.48\textwidth]{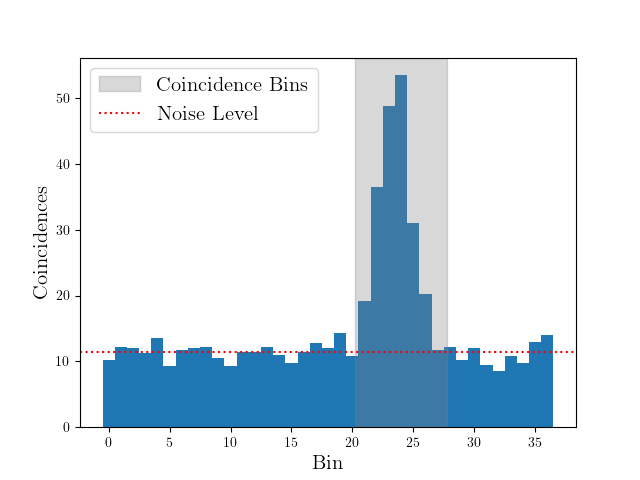}%
    }\hfill
    \subfloat[Cumulative distribution of coincidences for all measurements with a noise flux of \(70 \times 10^4\) photons per second. \label{fig: low noise cumc}]{%
        \includegraphics[width=0.48\textwidth]{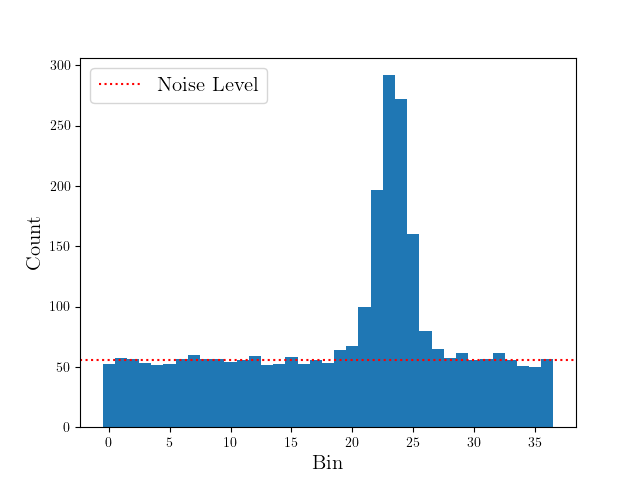}%
    }
    \subfloat[Cumulative distribution of coincidences for all measurements with a noise flux of \(16 \times 10^5\) photons per second. \label{fig: high noise cumc}]{%
        \includegraphics[width=0.48\textwidth]{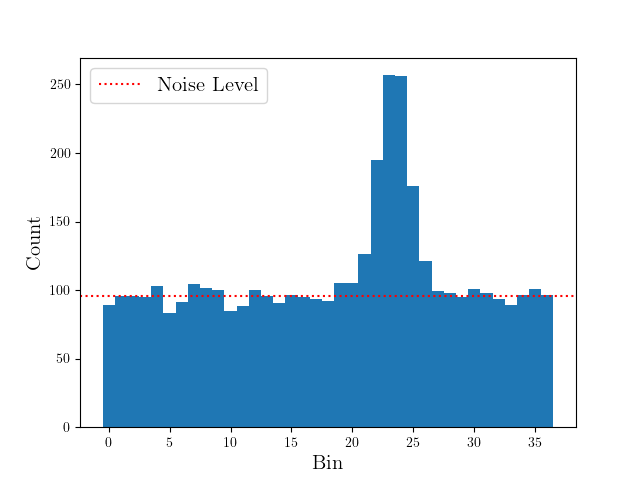}%
    }\hfill
    
    \vspace{1em} 

    \caption{Cumulative histograms for all measurements, and a histogram of coincidences for a single projective measurement. For each projective measurement onto an eigenstate, we used an 80-second integration time. Our coincidence histograms have a temporal resolution of 328 picoseconds, which we set as the bin size.}
    \label{fig: cumulative coincidences}
\end{figure*}

\end{document}